\documentclass[hidelinks,12pt]{amsart}
\usepackage[margin=1in]{geometry} 

\usepackage[utf8]{inputenc}
\usepackage[pdftex]{graphicx}
\usepackage{color}
\usepackage{float}
\usepackage{hyperref}
\usepackage{caption}
\usepackage{subcaption}

\usepackage{amsthm}

\DeclareMathOperator{\Tr}{Tr}

\newcommand{\WgO}{\mathrm{Wg}^{\mathrm{COE}}_{N,k}}
\newcommand{\WgU}{\mathrm{Wg}^{\mathrm{CUE}}_{N,k}}

\newcommand{\WgUk}[1]{\mathrm{Wg}^{\mathrm{CUE}}_{N,#1}}

\newcommand{\cF}{\mathcal{F}}
\newcommand{\setK}{\mathcal{K}}
\newcommand{\id}{{\mathrm{id}}}

\newcommand{\cc}[1]{{\overline{#1}}}
\newcommand{\Reg}[2]{\mathcal{R}_{#1}^{#2}} 
\newcommand{\Irreg}[2]{\mathcal{I}_{#1}^{#2}} 

\newcommand{\CUE}{\mathrm{CUE}}
\newcommand{\COE}{\mathrm{COE}}

\newtheorem{theorem}{Theorem}[section]
\newtheorem{lemma}[theorem]{Lemma}

\newtheorem{definition}[theorem]{Definition}

\theoremstyle{remark}
\newtheorem{remark}[theorem]{Remark}
\newtheorem{example}[theorem]{Example}

\begin{document}

\title{Convergence of moments of twisted COE matrices}

\author{Gregory Berkolaiko}
\address{Department of Mathematics, Texas A\&M University, College Station, TX 77843-3368, USA}
\email{berko@math.tamu.edu}

\author{Laura Booton}
\address{Department of Mathematics, Texas A\&M University, College Station, TX 77843-3368, USA}
\email{lbooton@math.tamu.edu}

\date{\today}

\begin{abstract}
  We investigate eigenvalue moments of matrices from Circular
  Orthogonal Ensemble multiplicatively perturbed by a permutation
  matrix.  More precisely we investigate variance of the sum of the
  eigenvalues raised to power $k$, for arbitrary but fixed $k$ and in
  the limit of large matrix size.  We find that when the permutation
  defining the perturbed ensemble has only long cycles, the answer is
  universal and approaches the corresponding moment of the Circular
  Unitary Ensemble with a particularly fast rate: the error is of
  order $1/N^3$ and the terms of orders $1/N$ and $1/N^2$ disappear
  due to cancellations.  We prove this rate of convergence using
  Weingarten calculus and classifying the contributing Weingarten
  functions first in terms of a graph model and then algebraically.
\end{abstract}

\maketitle

\section{Introduction}

Since their introduction by Dyson \cite{Dys_jmp62_i,Dys_jmp62_0}, the
so-called circular ensembles of random matrices have been used
successfully to model various physical processes, such as the time
evolution of complex quantum systems
\cite{Mehta_RandomMatrices,Haake_signatures} or scattering from a
potential of unknown or complex structure
\cite{BluSmi_prl90,beenakker97}, or even the local statistics of the
zeros of Riemann zeta and other $L$-functions \cite{MezSna_RMT_NT}.
The three classical circular ensembles --- unitary, orthogonal and
symplectic --- are used to model systems with different basic
symmetries.  In particular, the Circular Orthogonal Ensemble (COE)
corresponds to a system which is invariant with respect to time
reversal.

As well as modeling an entire system, a random matrix from COE can be
used to model scattering from a part of a composite system.  A
particular inspiration for the present paper is the work of Joyner,
M\"uller and Sieber \cite{JoyMulSie_epl14} where a small number of
random scatterers were combined with a carefully chosen magnetic
fluxes to produce a spectrum following Circular Symplectic (CSE)
statistics even though the system had no spin which is normally
associated with CSE.  This model was later realized experimentally
using microwave networks \cite{Reh+_prl16}.  For other examples of ``composite''
random matrix ensembles and their applications, see \cite{PozZycKus_jpa98}.

Mathematical analysis of such models would require understanding the
spectral statistics of products of random matrices.  In this paper we
make a step in this direction by studying the moments of the matrices
from COE multiplicatively perturbed by a fixed permutation matrix.
Intuitively, without the time-reversal symmetry which is destroyed by
the permutation, the result should follow the prediction of the
Circular Unitary Ensemble (CUE), even though the perturbed COE is but
a tiny submanifold\footnote{More precisely, the probability measure we
  consider is supported on a high codimension submanifold of $U(N)$,
  the compact manifold which is the support of the uniform probability
  measure defining the CUE.}  of the CUE.  We seek to confirm this
intuition with rigorous quantitative results and to find the
conditions on the permutation matrix sufficient for the convergence.

As another source of inspiration, we would like to cite the idea of
Degli Esposti and Knauf \cite{DegKna_jmp04}, who proposed that
periodic orbit expansions in quantum chaos should be modeled on
similar expansions in suitable random matrix ensembles.  It should be
noted that due to abundance of invariance and averaging available
within a classical random matrix model, there are normally more direct
methods to arrive at an answer for any particular quantity (such as
the spectral correlation functions or their Fourier transform, the
form factor), than to expand it into products of matrix elements and
do careful combinatorial accounting.  What Degli Esposti and Knauf
suggested is that doing the computation the hard combinatorial way in
random matrices would shed light on similar computations for quantum
chaotic systems where such expansions are often the best way to
proceed.  After the publication of \cite{DegKna_jmp04}, a lot of
progress was made in understanding the origin of the correct
combinatorial contributions
\cite{sr01,Sie02,BerSchWhi_jpa03,mulleretal04,heusleretal07,mulleretal09,BerKui_jmp13a,BerKui_jmp13b},
however a rigorous mathematical derivation with error control and
convergence analysis is still missing for any quantum chaotic model.

In the present work, we are firmly within the realm of random
matrices.  However, due to the perturbation considered, the invariance
of the random matrix ensemble (in the case of COE, invariance under
the conjugation $U \mapsto V^T U V$, with $V$ unitary) is broken and
expanding the moment into a product of matrix elements followed by
careful combinatorial accounting mentioned above becomes a necessity.
In addition to establishing convergence as mentioned above, we develop
a complete algebraic characterization of the permutations contributing
to the three leading orders of the $k$-th moment expansion.

\section{Summary of the main results}

Let $P_N$ be a fixed $N\times N$ permutation matrix.  We are
interested in the properties of the ``$P_N$-twisted COE'' ensemble, the
set of matrices of the form $P_NU$, where $U$ is a random unitary
symmetric matrix distributed according to Circular Orthogonal Ensemble
measure.  Circular Orthogonal Ensemble (COE) is the classical compact
symmetric space $U(N)/O(N)$ identified with the set of all unitary
symmetric matrices endowed with the unique probability measure
invariant under the action by the unitary group
\begin{equation}
  \label{eq:actionU}
  U \mapsto V U V^T, \qquad V\in U(N).
\end{equation}
Defining Circular Unitary Ensemble (CUE) as the compact group $U(N)$
with uniform (Haar) measure, we can represent COE as the image of CUE
under the mapping $V \mapsto V \cc{V^{-1}} = V V^T$; here the bar
denotes complex conjugation which plays the role of Cartan involution
in the case of $U(N)/O(N)$.  Practical aspects of integrating over the
CUE and COE are discussed in Appendix~\ref{sec:Weing_calc}.

We will investigate the $P_N$-twisted ensemble by studying the moments
\begin{equation}
  \label{eq:kmoment_def}
  M_k(N) := \left\langle \left| \Tr(P_N U)^k \right|^2
  \right\rangle_{\COE(N)}.
\end{equation}
We will study the asymptotics of $M_k(N)$ for an arbitrary fixed $k$
and for $N\to\infty$.  When changing the size of the matrices we have
to specify a new permutation matrix $P_N$ for every $N$.  It turns out
the answers depend mostly on the lengths of the cycles of $P_N$; our
results will be valid for arbitrary sequences of $P_N$ provided some
minimal cycle condition is satisfied.

For comparison purposes, we now state the moment formulas of the
classical Circular Unitary and Circular Orthogonal ensembles.
When $k<N$, we have
\begin{align}
  \label{eq:CUE_answer}
  \left\langle \left|\Tr U^k \right|^2 \right\rangle_{\CUE(N)}
  &=k, \\
  \label{eq:COE_answer}
  \left\langle \left| \Tr U^k \right|^2 \right\rangle_{\COE(N)}
  &= 2k - k\sum_{m=1}^k \frac{1}{m+(N-1)/2} = 2k + \mathcal{O}\left(\frac1N\right).
\end{align}
Both expressions can be derived from two-point correlation functions
calculated in \cite{Mehta_RandomMatrices}; explicitly they were
calculated in, for example, \cite[Eqs.~(45) and
(48)]{HaaKusSomSchZyc_jpa96}.  In addition, we refer to
\cite[Section~III]{DegKna_jmp04} for an evaluation of
\eqref{eq:CUE_answer} in terms of Weingarten functions, which is close
to the methods of the present work .  If the permutation matrix $P_N$
is an involution, i.e.\ $P_N^2=I$ or all cycles of $P_N$ are of length
1 or 2, the moments $M_k(N)$ are identical to the
$\COE(N)$ moments (see Lemma~\ref{lem:involution} below).

In order to get a sense of what happens for $P_N$ with long cycles, we
investigated a particular permutation matrix $P_N$ corresponding to
the grand cycle $(1\,2\,\ldots\,N)$.  This $P_N$ acts on a matrix $U$
with $N$ rows $r_j$ by cyclically shifting the rows as follows,
\begin{equation}
  \label{eq:Paction}
  P_N
  \begin{pmatrix}
    r_1 \\ r_2 \\ \vdots \\ r_N
  \end{pmatrix}
  =
  \begin{pmatrix}
    r_2 \\ r_3 \\ \vdots \\ r_1
  \end{pmatrix}.
\end{equation}

For low values of $k$ it is possible to use Weingarten calculus (see a
short review in Appendix~\ref{sec:Weing_calc}) and carefully enumerate
all contributing permutations to get exact results.  The enumeration
method behind this calculation will be reported elsewhere; here we
just list the results,
\begin{align}
  \label{eq:k2_result}
  M_2(N) &= 2, \\
  \label{eq:k3_result}
  M_3(N) &= \tfrac{3(N^3+3N^2-N-2)}{(N-1)(N+1)(N+3)}\\
         \nonumber &= 3 + \frac{3}{N^3} + \ldots,\\
  \label{eq:k4_result}
  M_4(N) &= 4, \\
  M_5(N) &= \tfrac{5(N^9+21N^8+129N^7+21N^6-2193N^5-55545N^4
           -3653N^3+14463N^2-4572N-9072)}
           {(N-3)(N-2)(N-1)(N+1)(N+2)(N+3)(N+5)(N+7)(N+9)}
           \label{eq:k5_result} \\ \nonumber
         &= 5 - \frac{1200}{N^3} + \ldots
\end{align}  

Three features stand out.  First, there is a convergence to the CUE
answer for the moments, cf.\ equation~\eqref{eq:CUE_answer}.  Second,
the convergence is unexpectedly fast: the terms of order $\frac1N$ and
$\frac1{N^2}$ are conspicuously absent.  Third, the answer matches CUE
exactly for $k=2,4$.  We can now confirm that the first two
observations are valid for all $k$ and a wide class of permutations
$P_N$.\footnote{We cannot yet provide any insight for the other
  natural question: is $M_k(N)=k$ for all even $k$?}

\begin{theorem}
  \label{thm:two_orders0}
  Let $P_N$ be an $N\times N$ permutation matrix such that all cycles
  of $P_N$ are longer than $2k$, $k \in \mathbb{N}$.  Then the $k$-th
  moment $M_k(N)$ defined by \eqref{eq:kmoment_def} is a rational
  function of $N$ independent of the particular choice of $P_N$ and
  satisfying the estimate
  \begin{equation}
    \label{eq:two_orders0}
    M_k(N) = k + \mathcal{O}\left(N^{-3}\right).
  \end{equation}
\end{theorem}

We remark here that the constant, implicit in the right-hand side of
equation~\eqref{eq:two_orders0}, is naturally independent of $N$ and
the particular choice of $P_N$ (since the function $M_k(N)$ is
independent of the latter), but is, in principle, dependent on $k$.
Extending this theorem to cover the possibility of $k$ increasing
together with $N$ remains a challenging question.

\subsection{Related results}

Of related works known to the authors, we would like to mention a
study by Po\'{z}niak, \.{Z}yczkowski and Ku\'{s}
\cite{PozZycKus_jpa98}, who surveyed several composite ensembles
obtained from the classical one by multiplying independently sampled
matrices.  While they did not consider the twisted COE ensemble we study
here, their physical motivation was similar.

In mathematical literature, there are several results concerning the
distribution of $\Tr(P_N M)$, which is shown to converge to normal as
$N\to\infty$ when $M$ is uniformly distributed in the orthogonal group
$O(N)$ \cite{DarDiaNew_incol03,Mec_tams08}.  The difference in the
ensemble (COE versus the orthogonal group) is probably irrelevant.
However in this work we address the so-called linear moments ---
expectations of $\left| \Tr(P_N U)^k \right|^2$ with $k$ potentially
large --- whereas studying the distribution of $\Tr(P_N U)$ is
equivalent to understanding the nonlinear moments --- expectations of
$\left| \Tr(P_N U) \right|^{2n}$.  To put it another way, we are
addressing the distribution of eigenvalues, whereas the results
of \cite{DarDiaNew_incol03,Mec_tams08} concern the normality of
individual entries of the matrix (this idea goes back to Borel
\cite{Bor_asens06}).

\subsection{Notation}
\label{sec:notation}

We now briefly review the notation used in the rest of the paper.
We will use notation $[N]$ and $[k]$ to denote the set of first $N$
(correspondingly, $k$) natural numbers.  As additional indexing symbols
we will use the natural numbers with bars on top and will denote
\begin{equation}
  \label{eq:kbar_def}
  [\cc{k}] = \left\{\cc{1},\cc{2}, \ldots, \cc{k}\right\}
\end{equation}
We will write $S_k$ for the symmetric group on $k$ elements.  We will
most actively use the symmetric group $S_{2k}$ which will be viewed as
the group of permutations of the ordered set
\begin{equation}
  \label{eq:setK}
  \setK := [k] \cup [\cc{k}]
  = \left[1,\cc{1},2, \cc{2},...,k, \cc{k} \right].
\end{equation}

To each permutation $\omega \in S_n$ we put in correspondence its
``cycle type'', a sequence of non-negative integers
$(\alpha_1, \alpha_2, \ldots)$ such that $\alpha_j$ is the number of
cycles of $\omega$ of length $j$.  Naturally, the sequence is
identically zero after some point and
$\sum_{j=1}^\infty j \alpha_j = n$.  We denote by $\ell(\omega)$ the
number of cycles of $\omega$, i.e.
\begin{equation}
  \label{eq:def_ell}
  \ell(\omega) = \sum_{j=1}^\infty \alpha_j.  
\end{equation}
We will normally record the cycle type as
$1^{\alpha_1} 2^{\alpha_2} 3^{\alpha_3} \cdots$, omitting all terms
with $\alpha_j=0$ and occasionally omitting the $1^{\alpha_1}$ term.

In regards to the permutation matrix $P_N$, we will abuse notation
slightly by also denoting the corresponding permutation by $P_N$,
occasionally dropping the subscript $N$ to reduce notational clutter.

Among the permutations in $S_{2k}$ we distinguish three permutations
we will use repeatedly,
\begin{align}
  \label{eq:Tdef}
  T &:= (1\,\cc1)(2\,\cc2)\cdots(k\,\cc{k}), \\
  \label{eq:Qdef}
  Q &:= (1\,\cc{k})(2\,\cc1)(3\,\cc2)\cdots(k\,\cc{k-1}),\\
  \label{eq:sdef}
  s &:= (1\,2\,\ldots k)(\cc1\,\cc2\,\ldots \cc{k}).
\end{align}
Note that $T$ and $Q$ are involutions, i.e.\ $Q^2 = T^2 = \id$.  We
will often need to refer to the cycle length of the commutators
$[\omega, T] := \omega T \omega^{-1} T^{-1}$ and
$[\omega, Q] := \omega Q \omega^{-1} Q^{-1}$, for some
$\omega \in S_{2k}$.  In his case we will drop the parentheses and
write simply $\ell[\omega,T]$ and $\ell[\omega,Q]$.

\section{Deriving the moments of twisted COE}

\subsection{The case of the involution}

We start with a simple result suggesting that short cycles in $P_N$
fail to change the eigenvalue statistics.

\begin{lemma}
  \label{lem:involution}
  Let $P$ be an orthogonal matrix such that $P^2=I$.  Then
  \begin{equation}
    \label{eq:involution}
    \left\langle \left| \Tr(P U)^k \right|^2 \right\rangle_{\COE(N)}
    = \left\langle \left| \Tr U^k \right|^2 \right\rangle_{\COE(N)}.
  \end{equation}
\end{lemma}

\begin{proof}
  The conditions on the matrix $P$ imply that $P^T = P^{-1} = P$,
  therefore $P$ is real symmetric and can be diagonalized as
  \begin{equation*}
    P = S \Lambda S^T,
  \end{equation*}
  where $S$ is an orthogonal matrix and $\Lambda$ is diagonal.  We can
  further represent $\Lambda = A A^T$, where $A$ is unitary (this
  representation is not unique).  Substituting this into the trace, we
  get
  \begin{align*}
    \Tr(P U)^k &= \Tr( S A A^T S^T U \cdots  S A A^T S^T U) \\
    &= \Tr(A^T S^T U \cdots  S A A^T S^T U S A)
    = \Tr(A^T S^T U S A)^k.
  \end{align*}
  But due to the invariance COE with respect to the action of the
  unitary group defined by \eqref{eq:actionU},
  \begin{equation*}
    \left\langle \left| \Tr(A^T S^T U S A)^k
      \right|^2 \right\rangle_{\COE(N)}
    = \left\langle \left| \Tr U^k \right|^2 \right\rangle_{\COE(N)},
  \end{equation*}
  since the integration measure is invariant with respect to change of
  variables from $U$ to $(SA)^T U SA$.
\end{proof}

\subsection{Trace expansion}

We begin our march towards the proof of Theorem~\ref{thm:two_orders0}
by expanding the trace of the power in the definition of moment
function $M_k(N)$,
\begin{align}
  \nonumber
  \Tr(PU)^k &= \sum_{n_1,n_2,\ldots, n_k=1}^{N} (PU)_{n_1,n_2}
  (PU)_{n_2,n_3} \cdots (PU)_{n_k,n_1}\\
  \label{eq:trace_expansion}
  &= \sum_{n_1,n_2,\ldots, n_k=1}^{N} U_{P(n_1),n_2}
    U_{P(n_2),n_3} \cdots  U_{P(n_k),n_1} \\
  \label{eq:usingF}
  &= \sum_{I \in \cF} U_{i_1, i_\cc1}
  U_{i_2,i_\cc2} \cdots U_{i_k,i_\cc{k}},
\end{align}
where in \eqref{eq:usingF} we assigned a
different letter to each index of $U$ and introduced the set $\cF$
\begin{equation}
  \label{eq:setF_def}
  \cF_N := \left\{
  (i_1, i_\cc1, i_2, \ldots,i_k,i_\cc{k}) \in [N]^{2k} \colon
   P_N(i_\cc{b}) = i_{b+1} \quad \forall b\in [k]
  \right\},
\end{equation}
with $b+1$ understood to be modulo $k$.  The ordered $2k$-tuple
$(i_1, i_\cc1, i_2, \ldots,i_k,i_\cc{k})$ is denoted by $I$ and can be
viewed as a function from $\setK$ to $[N]$ (refer to
Section~\ref{sec:notation} for notation).


Expanding the square in the definition of $M_k(N)$, we obtain
\begin{equation}
  \label{eq:kmoment_expand1}
  M_k(N) = \sum_{I,J \in \cF_N}
  \left\langle
    U_{i_1, i_\cc1} U_{i_2,i_\cc2}
    \cdots U_{i_k,i_\cc{k}}
    \cc{U}_{j_1, j_\cc1} \cc{U}_{j_2,j_\cc2} \cdots
    \cc{U}_{j_k,j_\cc{k}} 
  \right\rangle_{\COE(N)},
\end{equation}
where the bar over $U$ denotes the complex conjugation.  The average
of the product of COE matrix elements is given by the so-called
Weingarten calculus \cite{BroBee_jmp96,Mat_rmta12,ColMat_alea17},
which we review in Appendix~\ref{sec:summary_COE} (see also
Appendix~\ref{sec:summary_CUE} for its CUE counterpart).  In
particular, it is zero unless index values $J$ are a permutation of
the index values $I$.  More precisely,
\begin{equation}
  \label{eq:kmoment_expand2}
  M_k(N) = 
  \sum_{I,J \in \cF_N} \,\,
  \sum_{\omega \in S_{2k} \colon I\circ\omega = J} \WgO(\omega).
\end{equation}
where $\WgO(\omega)$ is the Weingarten function of COE.  The function $\WgO$
depends only on the conjugacy class (equivalently, ``cycle type'' or
the vector of cycle lengths) of the permutation
$[\omega,T] := \omega T\omega^{-1} T^{-1}$, where $T \in S_{2k}$ is the
fixed permutation defined in equation \eqref{eq:Tdef}.

Switching the order of summation we arrive to
\begin{equation}
  \label{eq:kmoment_expand}
  M_k(N) = 
  \sum_{\omega \in S_{2k}}
  \WgO(\omega) 
  \sum_{I \in \cF_N \colon I\circ\omega \in \cF_N} 1
  =: \sum_{\omega \in S_{2k}} \WgO(\omega) \Phi(\omega).
\end{equation}

The following theorem gives a simple description of the function
$\Phi(\omega)$ introduced in \eqref{eq:kmoment_expand} and is a
major step towards proving the main result.

\begin{theorem}
  \label{thm:countingIJF}
  Recall the definition of the set $\cF_N$,
  equation~\eqref{eq:setF_def}.  For any $\omega \in S_{2k}$,
  \begin{align}
    \label{eq:Phi_definition}
    \Phi(\omega)
    &:= \#\left\{I \in \cF_N \colon I\circ\omega \in \cF_N\right\} \\
    \label{eq:Phi_formula}
    &= \chi_\omega N^{\frac12 \ell[\omega,Q]},
  \end{align}
  where the factor $\chi_\omega$ is the indicator function of the set
  \begin{equation}
    \Omega_{k,N}:=\{\omega\in S_{2k}: \exists\, I\in [N]^{2k}
    \mbox{ s.t. } I\circ\omega\in \cF_N \},
  \end{equation}
  and $\ell[\omega,Q]$ denotes the number of cycles in the commutator
  $\omega Q\omega^{-1} Q^{-1}$ with $Q$ defined in
  equation~\eqref{eq:Qdef}.  The set $\Omega_{k,N}$, and therefore its
  indicator $\chi_\omega$ and the function $\Phi(\omega)$, are
  independent of $N$ and permutation $P_N$ as long as all of cycles of
  $P_N$ are longer than $2k$.
\end{theorem}

We will prove Theorem~\ref{thm:countingIJF} in the next section.  At
this point we would like to point out that
Theorem~\ref{thm:countingIJF} effectively establishes the first claim
of our main result, Theorem~\ref{thm:two_orders0}, that $M_k(N)$ is
independent of $P_N$: equation~\eqref{eq:kmoment_expand} represents
$M_k(N)$ as a finite sum of terms depending on $\omega$ only.

\subsection{Graph model of a permutation}
\label{sec:graph_model}

To each permutation $\omega \in S_{2k}$ we now associate a multigraph
$G_\omega$ which will enable us to easily access all the data needed
to evaluate the contribution of a permutation $\omega$ to the sum in
\eqref{eq:kmoment_expand}.

\begin{definition}
  \label{def:graph_model}
  The \emph{graph model} of $\omega \in S_{2k}$ is a
  multigraph\footnote{A \emph{multigraph} can have more than one edge
    between a given pair of vertices.}
  $G_\omega$ with the vertex set $\setK$, see~\eqref{eq:setK}, and the
  following labeled edges: for every $b \in [k]$ there is
  \begin{itemize}
  \item an undirected \emph{solid} edge between $b$ and $\cc{b}$,
  \item a directed \emph{solid} edge from $\cc{b}$ to $b+1$ (modulo $k$),
  \item an undirected \emph{dashed} edge between $\omega(b)$ and
    $\omega(\cc{b})$,
  \item a directed \emph{dashed} edge from $\omega(\cc{b})$ to
    $\omega(b+1)$.
  \end{itemize}
\end{definition}

\begin{figure}
  \centering
  \begin{subfigure}[b]{0.4\textwidth}
    \centering
    \includegraphics[scale=0.6]{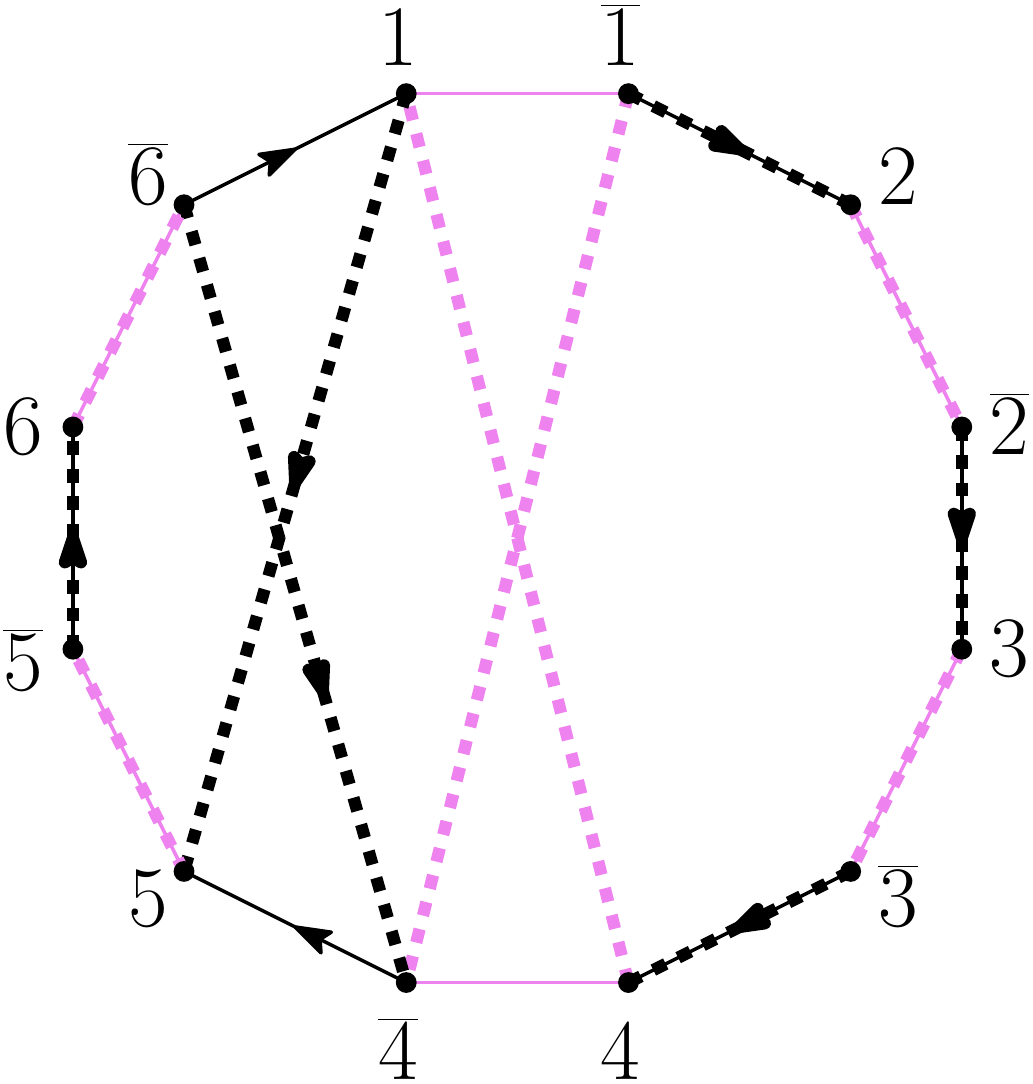}
    \caption{$\omega = (1\, \cc4)$}
    \label{fig:graph_model_14}
  \end{subfigure}
  \qquad
  \begin{subfigure}[b]{0.4\textwidth}
    \centering
    \includegraphics[scale=0.6]{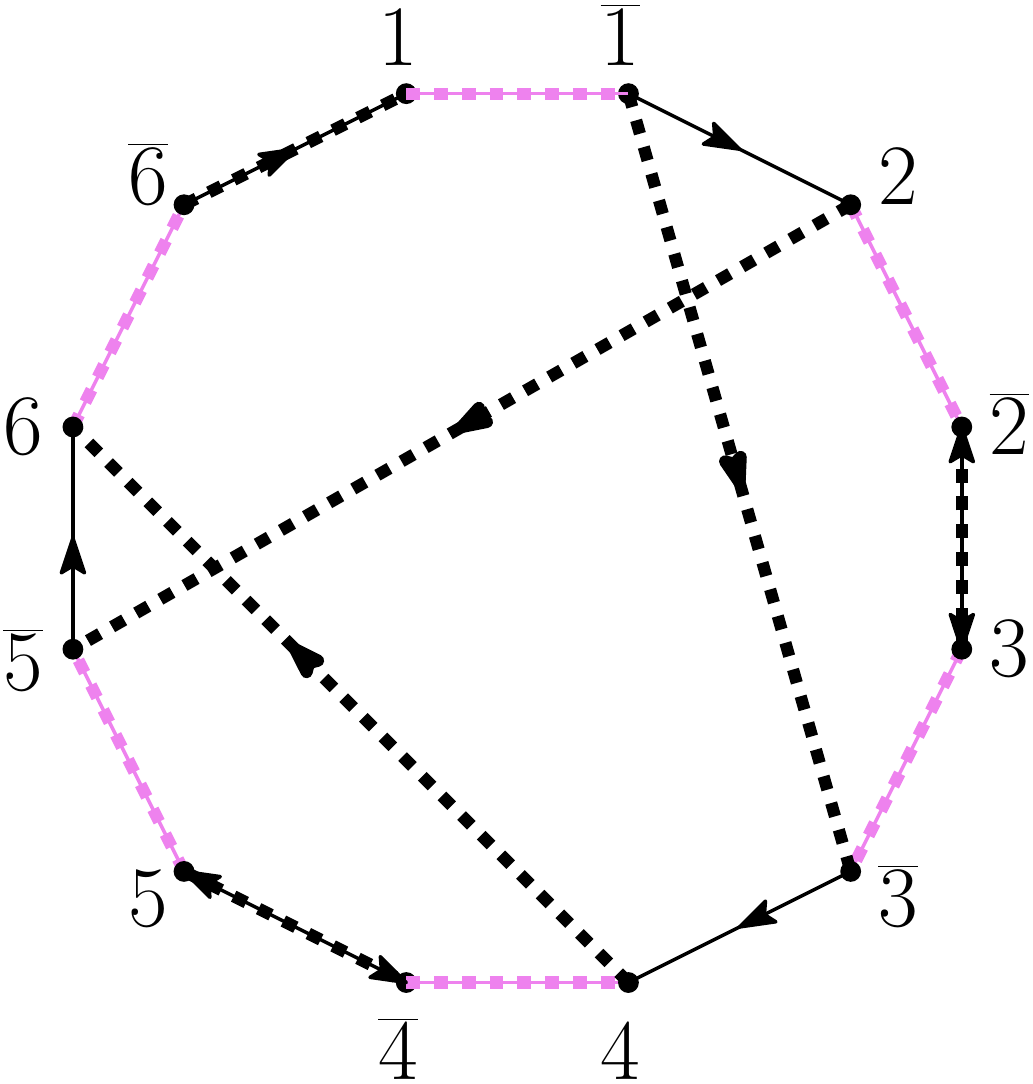}
    \caption{$\omega = (2\, \cc3)(\cc2 \, 3)(4 \, \cc5)(\cc4 \, 5)$}
    \label{fig:graph_unbalanced}
  \end{subfigure}
  \caption{Examples of graph models with $k=6$.  Undirected edges are
    drawn in violet.  When there is more than one edge between a pair
    of vertices (one dashed and one solid), they are drawn on top of
    each other; if they are directed in the same way only one arrow is
    drawn.}
\label{fig:graph_model_examples}
\end{figure}

Some examples of graph models are shown in
Figure~\ref{fig:graph_model_examples}.  Each vertex of $G_\omega$ has
degree 4, being incident to one each of the four types of edges
(directed edges may be incoming or outgoing).  More specifically, each
vertex $z\in \mathcal{K}$ is incident to the solid undirected edge
connecting $z$ and $\cc z$, which we will denote by
$(z-\overline{z})$, and is incident to the dashed undirected edge
$(z-\omega T \omega^{-1}(z))$. The directed solid edge will be
$(Q(z) \to z)$ if $z\in [k]$ and $(z \to Q(z))$ if $z\in [\cc k]$. The
directed dashed edge will be $(\omega Q \omega^{-1}(z) \to z)$ if
$\omega^{-1}(z)\in [k]$ or $(z \to \omega Q \omega^{-1}(z))$ if
$\omega^{-1}(z)\in [\cc k]$.

The solid edges of $G_\omega$ form a cycle graph $C_s$ with the same
vertex set $\setK$ and the edges
\begin{equation}
  \label{eq:cycle_s}
  \to1-\cc 1 \to 2- ...-\cc k \to .
\end{equation}
The dashed edges similarly form a cycle graph $C_d$ with the edges
\begin{equation}
  \label{eq:cycle_d}
  \to \omega(1) - \omega(\cc 1) \to \omega(2)-...- \omega(\cc k) \to,
\end{equation}
and $\omega$ provides an isomorphism between $C_s$ and $C_d$.  Note
that all directed edges point in the same direction in cycle graphs
$C_s$ and $C_d$.

\begin{figure}
  \centering
  \begin{subfigure}[b]{0.4\textwidth}
    \centering
    \includegraphics[scale=0.6]{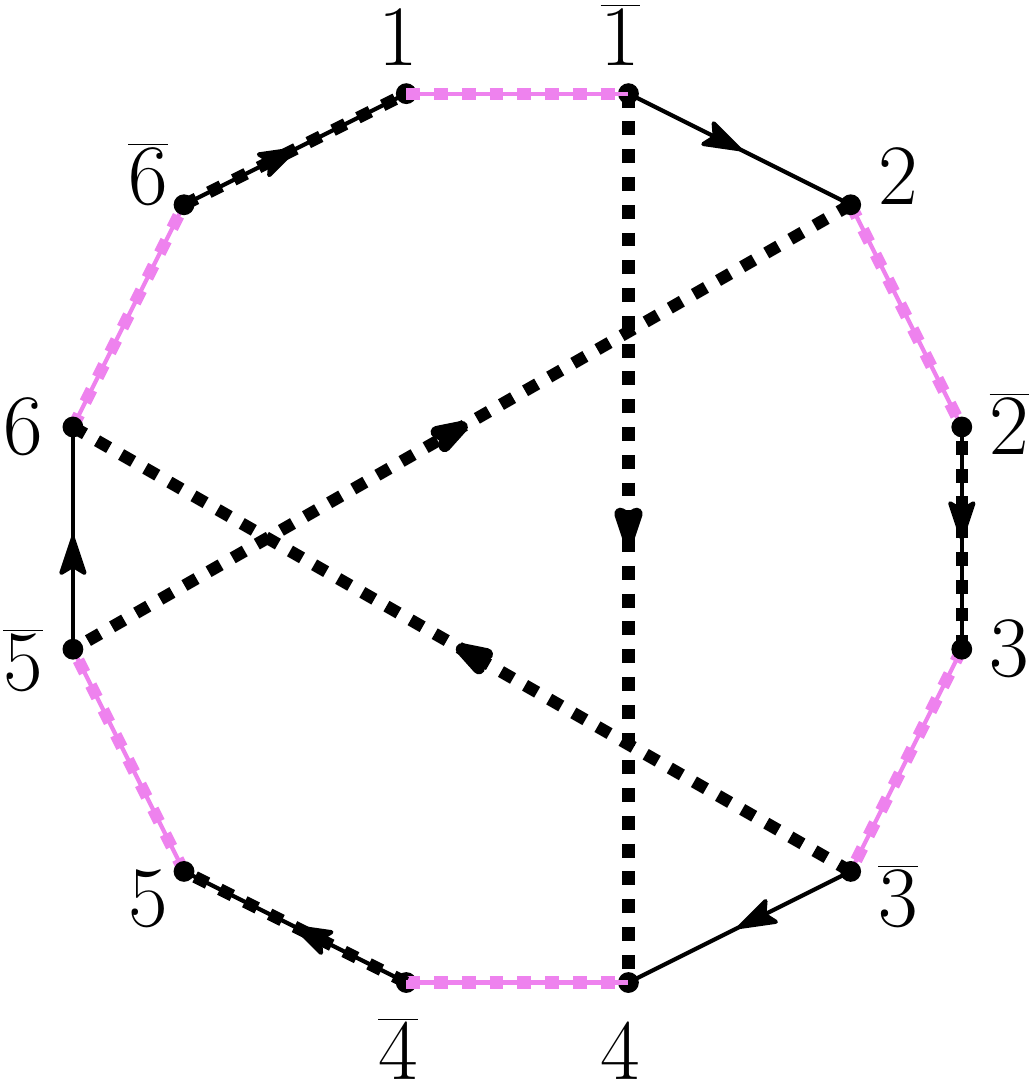}
    \caption{$\omega = (2\, 4)(3 \, 5)(\cc2 \, \cc4)(\cc3 \, \cc5)$}
    \label{fig:graph_balanced1}
  \end{subfigure}
  \qquad
  \begin{subfigure}[b]{0.4\textwidth}
    \centering
    \includegraphics[scale=0.6]{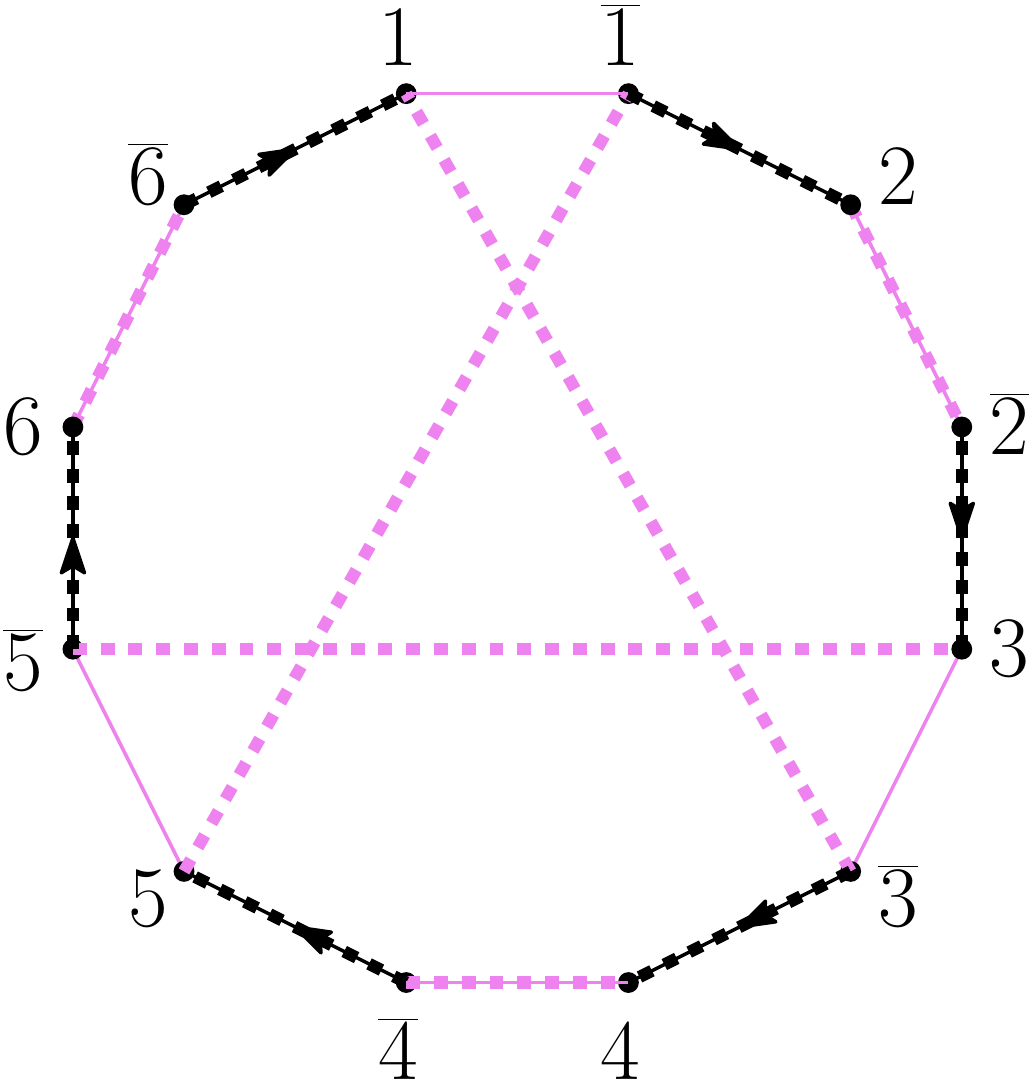}  
    \caption{$\omega = (2\,4)(3\,5)(\cc{1}\,\cc{3})(\cc{2}\,\cc{4})$}
    \label{fig:graph_balanced2}
  \end{subfigure}
  \caption{Some graph models of permutations with balanced
    cycles, $k=6$.}
\label{fig:balanced_cycles}
\end{figure}

In a slight abuse of usual graph terminology, we will use the term
\emph{directed cycle} to refer to a cycle of $G_\omega$ that is made
entirely of directed edges, but without regard for their direction.  A
directed cycle will be called \emph{balanced} if there is equal number
of edges going into each direction.  See
Figure~\ref{fig:balanced_cycles} for two examples of graphs with
balanced cycles and Figure~\ref{fig:graph_unbalanced} for an
unbalanced example (namely, cycles $\{\cc2, 3\}$ and $\{\cc4, 5\}$).

\begin{lemma}
  \label{lem:counting_cycles}
  Recall $T = (1\,\cc1)\ldots (k\,\cc{k})$,
  $Q = (\cc{1}\,2)(\cc{2}\,3)\ldots(\cc{k}\,1)$ and
  $s = (1\,2\,\ldots k)(\cc1\,\cc2\,\ldots \cc{k})$.
  For each $\omega\in S_{2k}$ the graph $G_\omega$ has the following
  properties.
  \begin{enumerate}
  \item The subgraph of $G_\omega$ consisting of all undirected edges
    is a disjoint union of cycle graphs; each cycle is of even length
    and alternates solid and dashed edges.  The number of cycles of
    length $2n$ is half the number of 
    $n$-cycles
    of the permutation $[\omega,T] :=\omega T\omega^{-1} T^{-1}$.
  \item \label{item:dir_cycles} The subgraph of $G_\omega$ consisting
    of all directed edges is a disjoint union of cycle graphs
    (ignoring the direction of the edges); each cycle is of even
    length and alternates solid and dashed directed edges.  The number
    of cycles of length $2n$ is half the number of $n$-cycles of the
    permutation $[\omega,Q] := \omega Q \omega^{-1} Q^{-1}$.
  \item \label{item:graphs isomorphic} The graphs of $G_\omega$ and
    $G_{\omega'}$ are identical if and only if $\omega'=\omega s^n$
    for any $n\in\mathbb{Z}$; the graphs $G_\omega$ and $G_{ \omega'}$
    are isomorphic if and only if $\omega'=s^n \omega s^m$ for some
    $n,m\in\mathbb{Z}$.
  \item \label{item:balanced} If all cycles of $P_N$ are longer than
    $2k$, there exist $I\in\cF_N$ such that $I\circ\omega\in \cF_N$ if
    and only if all directed cycles of $G_\omega$ are balanced.    
  \end{enumerate}
\end{lemma}

\begin{proof}  
  Since each vertex of $G_\omega$ is incident to exactly one each of
  the four types of edges, removing all directed edges we obtain a
  graph with each vertex of degree 2, incident to one solid and one
  dashed edge.  Such a graph must be a union of cycles.  It has edge
  coloring with two colors so it must be bipartite and therefore all
  cycles are even.

  Consider a cycle of length $2n$ and let $(b,\cc b)$ be one of its
  solid edges.  Starting at vertex
  $b$, we will traverse the cycle by applying permutations.  We know that
  $T^{-1}(b)=\cc b$, so we have traversed the first solid edge. Next
  by applying $\omega T \omega^{-1}$ to $\cc b$ we traverse the
  undirected dashed edge incident to $\cc b$. Hence, by applying
  $[\omega,T]=\omega T \omega^{-1} T^{-1}$ to a vertex $b$, we
  traverse two edges of its undirected cycle. Since $b$ is in a graph
  cycle of length $2n$, we must have $[\omega,T]^n(b)=b$. By applying
  $[\omega,T]$ to $\cc b$ $n$ times, we will traverse the other $n$
  vertices of our graph cycle, again returning to $\cc b$. Thus we
  have found two $n$-cycles of $[\omega,T]$ that directly correspond
  to the vertices of a $2n$-cycle in $G_{\omega}$.

  We may traverse directed $2n$-cycles of $G_{\omega}$ in a
  similar manner using $[\omega,Q]$, again finding two corresponding
  $n$-cycles of $[\omega,Q]$.

  Consider the graphs $G_{\omega}$ and $G_{\omega s^n}$ for some
  permutation $\omega$. The solid subgraph $C_s$ will be the same for
  both graphs by definition. The dashed subgraph $C_d$ in $G_{\omega s^n}$
  can then be represented by
  \begin{align}
    &\hspace{4mm} \to \omega s^n(1)-\omega s^n(\cc 1) \to \omega
      s^n(2)- \ldots -\omega s^n(\cc k) \to\\
    \label{eqn:omega s^n path}
    &=\hspace{1mm}\to \omega(1+n)-\omega(\cc {1+n}) \to \omega(2+n)-
      \ldots - \omega(\cc {k+n}) \to.
  \end{align}
  Since addition is performed modulo $k$, we see that \eqref{eqn:omega
    s^n path} is identical to \eqref{eq:cycle_d}.  Thus the graphs
  $G_{\omega}$ and $G_{\omega s^n}$ are the same as unions of
  identical pairs of graphs.
     
  Conversely, consider two graphs $G_{\omega}$ and $G_{\omega'}$ that
  are the same. Then their dashed subgraphs will be the same, that is:
  \begin{align}
    & \hspace{4mm} \to \omega(1) - \omega(\cc1) \to \omega(2) -
      \ldots -  \omega(\cc k) \to \\
    &=  \to \omega'(1) - \omega'(\cc1) \to \omega'(2) -
      \ldots -  \omega'(\cc k) \to.
  \end{align}
  This implies that $\omega(i+n)=\omega'(i)$ for all $i\in [k]$ and
  for some $n$. Thus $\omega s^n=\omega'$.
  
  The isomorphism between $G_\omega$ and $G_{s^n \omega}$ is given by
  $s^n$.  For example, a solid edge between $b$ and $\cc{b}$ in $G_\omega$ is
  mapped to the solid edge between $(b+n)$ and  $\cc{b+n})$ in $G_{s^n \omega}$.
  A dashed directed edge from $\omega(\cc{b})$ to
  $\omega(b+1)$ in $G_\omega$ is mapped to a dashed edge from
  $s^n\omega(\cc{b})$ to $s^n\omega(b+1)$ in $G_{s^n \omega}$ and
  similarly for the other types of edges.
  
  Conversely, if two graphs $G_{\omega}$ and $G_{\omega'}$ are
  isomorphic, the solid subgraph $C_s$ must remain unchanged and the
  only possible isomorphisms are rotations $s^n$.  Thus for some
  $s^n$, $G_{s^n \omega}$ is the same as $G_{\omega'}$. Hence we know
  that $\omega'=s^n\omega s^m$ for some $n,m\in [k]$.

  To establish part~(\ref{item:balanced}) of the Lemma, we discuss the
  meaning of directed edges. Comparing the definition of a solid edge
  and the definition of $\cF_N$, equation~\eqref{eq:setF_def}, we see
  that $I\in \cF_N$ if and only if $P(i_u)=i_v$ for every solid
  directed edge $(u\to v)$.  Similarly, $I\circ\omega \in \cF_N$ if
  and only if $P(i_u)=i_v$ for every dashed directed edge $(u\to v)$.
  To see the latter, observe that a dashed directed edge $(u\to v)$
  means there exists a $b\in[k]$ such that $u = \omega(\cc{b})$ and
  $v=\omega(b+1)$ and, on the other hand, $I\circ\omega \in \cF_N$
  means that $P_N\left(i_{\omega(\cc{b})}\right)=i_{\omega(b+1)}$.

  Let us now consider a directed cycle.  Labeling the vertices in the
  cycle $v_1,\ldots,v_{2n}$, we obtain $2n$ equations of the form
  $i_{v_{j+1}}=P^{\epsilon_j}(i_{v_{j}})$, where $\epsilon_j=\pm1$
  depending on the direction of the edge, $+1$ for $(v_j\to v_{j+1})$
  and $-1$ for $(v_j \leftarrow v_{j+1})$. Continuing the substitution
  process all the way around the cycle, we will eventually obtain an
  equation relating $i_{v_1}$ to itself of the form
  $i_{v_1}=P^m\left(i_{v_1}\right)$ where
  $m=\sum_{j=1}^{2n}\epsilon_j$. Because $|m|\leq 2n \leq 2k$ and we
  are are only allowing $P$ to have cycles strictly longer than $2k$,
  one can make an assignment to $i_{v_1}$ if and only if $m=0$.
  Observe that $m=0$ can only occur if there are exactly $n$ edges
  directed along the cycle and $n$ edges pointing in the other
  direction, i.e. the cycle is balanced.

  For the later use, we also note that if the cycle is balanced,
  $i_{v_1}$ can be chosen arbitrarily and this choice uniquely
  determines the values of $i_{v_2}, \ldots, i_{v_{2n}}$ through the
  recursion $i_{v_{j+1}}=P^{\epsilon_j}(i_{v_{j}})$.  
\end{proof}

\begin{remark}
  It is not hard to describe the set of all ``allowable'' graphs that
  are graph models of some permutation $\omega\in S_{2k}$ and to show that
  the mapping between $S_{2k}$ and allowable graphs is $k$-to-$1$, but
  we will not need it in this paper.
\end{remark}

\begin{remark}
  There is also an algebraic characterization of $\chi(\omega)$ that
  we present here without proof since we will not use it.  For any
  $z\in \mathcal{K}:=[k]\cup[\cc k]$ denote by $O_z$ the orbit of $z$
  under the action of $[\omega,Q]:=\omega Q \omega^{-1} Q^{-1}$. Then
  $\chi(\omega)=1$ if and only if for all $z\in\mathcal{K}$ the sets
  $O_z$ and $\omega^{-1}(O_z)$ contain the same number of symbols from
  $[\cc k]$.
\end{remark}

We recall now that in expansion \eqref{eq:kmoment_expand} of $M_k(N)$,
the Weingarten function depends on the cycle type of the permutation
$[\omega,T]$, whereas the factor $\Phi(\omega)$ depends on the cycle
type of the permutation $[\omega,Q]$ as well as the function
$\chi_\omega$.  Now that this information is represented
readily in the graph model $G_\omega$, we are ready to prove
Theorem~\ref{thm:countingIJF}.

\begin{proof}[ Proof of Theorem~\ref{thm:countingIJF}]
  As long as $P_N$ has cycles longer than $2k$, a permutation
  $\omega\in S_{2k}$ belongs to the set $\Omega_{k,N}$ if and only if
  the graph $G_\omega$ has balanced directed cycles.  The latter
  property is independent of $N$ and the particular choice of $P_N$.

  To count the number of ordered $2k$-tuples $I$ appearing in the
  definition of the function $\Phi(\omega)$,
  equation~\eqref{eq:Phi_definition}, we note that for each directed
  cycle we have the freedom to chose $i_{v_1}$ arbitrarily from $[N]$;
  all other $i_{v}$ participating in the cycle are determined
  recursively (see the proof of Lemma~\ref{lem:counting_cycles}).  By
  Lemma \ref{lem:counting_cycles}, part~(\ref{item:dir_cycles}), the
  number of directed cycles is $\frac{1}{2}\ell[\omega,Q]$.  Since
  each directed cycle gives $N$ choices independently of other cycles,
  the total number of possible choices of $I$ is
  $N^{\frac12\ell[\omega,Q]}$.
\end{proof}

\subsection{Contributing permutations}
\label{sec:counting_perm}

We are now getting ready to use expansion \eqref{eq:kmoment_expand} to
evaluate the $k$-th moment $M_k(N)$.  We review the properties of
Weingarten COE function $\WgO$ in Appendix~\ref{sec:summary_COE}.  At
this point we just need the observation that
\begin{equation*}
  \WgO(\omega) = \mathcal{O}\left(\frac1{N^{2k-\frac12\ell[\omega,T]}}\right),
\end{equation*}
where $\ell[\omega,T]$ is the length (number of cycles) of the
permutation $[\omega,T]$.

Combining this with Theorem~\ref{thm:countingIJF}, we obtain that the
contribution of a permutation $\omega$ has the leading order of $1/N$
to power $2k-\frac12\ell[\omega,T] - \frac12\ell[\omega,Q]$.

\begin{example}
  If $\omega = s^n$, then $\frac12\ell[\omega,T] =
  \frac12\ell[\omega,Q] = k$ producing a constant term contribution.
  We will show that all other cycles give contribution with a higher
  power of $1/N$.
\end{example}

To establish Theorem~\ref{thm:two_orders0} we will need to investigate
graphs with large $\frac12\ell[\omega,T] + \frac12\ell[\omega,Q]$.
This means that vast majority of directed and undirected cycles in the
graph models $G_\omega$ ought to have length $2$ (to maximize the
number of cycles that can be made from a fixed number of edges); other
cycles we will call ``long''.  In fact we will only need to consider
graph models with
\begin{itemize}
\item no long cycles --- leading order
\item one long cycle (of either type)
  of length 4 --- order $1/N$
\item two cycles of length 4 or one cycle of length 6 --- order $1/N^2$.
\end{itemize}

We further classify the permutations $\omega$ by whether or not they
mix the elements of $\setK$ of different types, with and without bars.

\begin{definition}
  \label{def:RegIrreg}
  A permutation $\omega \in S_{2k}$ is called \emph{regular} if, for
  every $b \in [k] \subset \setK$, $\omega(b) \in [k]$.  The set of
  all regular permutations can be naturally represented as
  $S_k \times S_k \subset S_{2k}$, with two halves acting on $[k]$ and
  $[\cc{k}]$ respectively.  A permutation $\omega \in S_{2k}$ that is
  not regular (i.e.\ there exist $b_1, b_2 \in [k]$ such that
  $\omega(b_1) = \cc{b}_2$) is called \emph{irregular}.

  We will denote by $\Reg{\alpha}{\beta}$ (corresp.\
  $\Irreg{\alpha}{\beta}$) the set of regular (corresp.\ irregular)
  permutations $\omega\in S_{2k}$ with $\chi_\omega = 1$, with cycle
  half-type of $[\omega,T]$ equal to $\alpha$ and with cycle half-type
  of $[\omega,Q]$ equal to $\beta$.
\end{definition}

\begin{example}
  \label{ex:regular}
  The regular permutation
  $\omega = (2\,4)(3\,5)(\cc{2}\,\cc{4})(\cc{3}\,\cc{5})$ contributes
  $\Reg{\id}{3^1}$, since $G_\omega$ has a directed cycle of length 6,
  see Figure~\ref{fig:graph_balanced1}.  The regular permutation
  $\omega = (2\,4)(3\,5)(\cc{1}\,\cc{3})(\cc{2}\,\cc{4})$ contributes
  to $\Reg{3^1}{\id}$ since its cycle of length $6$ is undirected, see
  Figure~\ref{fig:graph_balanced2}.  Both contribute to $M_k(N)$ at
  order $1/N^2$.

  Permutations shown in Figure~\ref{fig:graph_model_examples} are
  irregular.  The permutation $\omega = (1\, \cc4)$ in
  Figure~\ref{fig:graph_model_14} contributes to $\Irreg{2^1}{2^1}$
  while the permutation in Figure~\ref{fig:graph_unbalanced} has
  unbalanced cycles and hence does not contribute.
\end{example}

\begin{remark}
  \label{rem:reg_always_contribute}
  It can be easily seen from the definition of the graph model
  $G_\omega$ that if $\omega$ is regular then both types of directed
  edges point from a vertex in $[\cc{k}]$ to a vertex in $[k]$.
  Therefore, along any directed cycle the edge directions must
  alternate and, by Lemma~\ref{lem:counting_cycles} part
  (\ref{item:balanced}), $\chi_\omega=1$ for every regular $\omega$.
\end{remark}

We will classify all \emph{irregular} permutations contributing to the
remaining relevant sets, such as $\Irreg{2^1}{\id}$,
$\Irreg{\id}{2^1}$, $\Irreg{2^1}{2^1}$ etc.  We will \emph{not},
however, classify regular permutations directly.  Instead, in the
course of proving Theorem~\ref{thm:two_orders0} we will establish a
pair of identities which turns out to be sufficient for evaluating
their contribution.

We now list the necessary information about permutations contributing
to the expansion of $M_k(N)$ to three leading orders.  The proofs are
highly technical and are deferred to
Appendices~\ref{sec:counting_regular} and \ref{sec:class_irreg}
correspondingly.

\begin{lemma}
  \label{lem:regular_perm}
  We have
  \begin{align}
    \label{eq:shifts}
    &\Reg{\id}{\id} = \left\{s^n\right\}_{n\in[k]}
    \qquad \qquad
    \left|\Reg{\id}{\id}\right| = k,\\
    \label{eq:alphas}
    &\Reg{\id}{2^1} = \emptyset.
  \end{align}
\end{lemma}


\begin{lemma}
  \label{lem:irregular_perm}
  We have
  \begin{align}
    \label{eq:irrega1}
    \Irreg{\alpha}{\id}
    &= \emptyset
    \quad \mbox{for any }\alpha, \\
    \label{eq:irreg12}
    \Irreg{\id}{2^1}
    &= \left\{s^n (1\,\cc{1}) s^m \right\}_{n,m \in [k]}
    &&\left|\Irreg{\id}{2^1}\right| = k^2,
    \\
    \label{eq:irreg22}
    \Irreg{2^1}{2^1}
    &= \left\{s^n (b\,\cc{1}) s^m\right\}_{n,m \in [k],\
      3 \leq b \leq k}
    &&\left|\Irreg{2^1}{2^1}\right| = k^2(k-2),
    \\  
    \label{eq:irreg14}
    \Irreg{\id}{2^2}
    &= \left\{s^n (1\,\cc{b})(b\,\cc{1}) s^m,\
      s^n (1\,\cc1)(b\,\cc{b}) s^m
      \right\}_{n,m \in [k],\
      3 \leq b \leq k-1}
    &&\left|\Irreg{\id}{2^2}\right| = k^2(k-3),
    \\
    \label{eq:irreg13}
    \Irreg{\id}{3^1}
    &= \left\{s^n (1\,\cc1)(2\,\cc{2}) s^m\right\}_{n,m \in [k]}
    &&\left|\Irreg{\id}{3^1}\right| = k^2.
  \end{align}
\end{lemma}

\begin{remark}
  The restriction $b \geq 3$ in equations \eqref{eq:irreg22} and
  \eqref{eq:irreg14} suggest a natural question about which set do
  $\omega = (2\,\cc1)$ and $\omega = (1\,\cc2)(2\,\cc1)$ belong to.
  It turns out they do not contribute at all due to having unbalanced
  directed cycles.
\end{remark}

\begin{remark}
  Counting directly from our set notation, it would appear that
  $\big|\Irreg{\id}{2^2}\big|=2k^2(k-3)$. However, there is a symmetry in this
  notation causing each permutation to be counted exactly twice.  The
  details are described in the proof of \eqref{eq:irreg22}, see
  Lemma~\ref{lem:irreg22}.
\end{remark}


\subsection{Proof of the main result}
\label{sec:main_proof}

The last ingredient we need for proving our main result is the
expansion of the relevant Weingarten functions in inverse powers of
$N$.  These expansions are obtained using orthogonality relations in
Appendix~\ref{sec:summary_COE}.

\begin{proof}[Proof of Theorem \ref{thm:two_orders0}]
  Recall that through a combination of
  equation~\eqref{eq:kmoment_expand} and Theorem~\ref{thm:countingIJF}
  we have represented $M_k(N)$ as
  \begin{equation*}
    M_k(N)
    = \sum_{\omega \in S_{2k}} \chi_\omega \WgO(\omega)
    N^{\frac12 \ell[\omega,Q]},
  \end{equation*}
  where the indicator function $\chi_\omega$ was shown to be
  independent of the permutation $P_N$ under the conditions of the
  present theorem.  We thus have a finite sum of rational functions of
  $N$ with no dependence on $P_N$ in any contributing term.  Therefore
  $M_k(N)$ is a rational function of $N$ and has a convergent
  asymptotic expansion in $\frac1N$.  We are therefore justified (for
  a fixed $k$) to evaluate this expansion term-by-term.
  
  We start by evaluating the corresponding moments for the CUE
  matrices.  This will allow us to evaluate the contribution of the
  regular permutation to $M_k(N)$ from the limited information
  provided by Lemma~\ref{lem:regular_perm}; in addition, it is a
  warm-up for the more lengthy calculation of the full COE expansion.
  Define
  \begin{equation}
    \label{eq:M_CUE_def}
    \mathcal{M}^{\mathrm{CUE}}_k(N)
    := \left\langle \left| \Tr(PU)^k \right|^2
    \right\rangle_{\mathrm{CUE}(N)}.
  \end{equation}
  On one hand, since the measure on $\CUE(N)$ is by definition
  invariant with respect to multiplication by a unitary matrix $P$, we
  have
  \begin{equation}
    \label{eq:M_CUE_eval}
    \mathcal{M}^{\mathrm{CUE}}_k(N)
    = \left\langle \left| \Tr(PU)^k \right|^2
    \right\rangle_{\mathrm{CUE}(N)}
    = \left\langle \left| \Tr(U)^k \right|^2
    \right\rangle_{\mathrm{CUE}(N)}
    = k,
  \end{equation}
  see equation~\eqref{eq:CUE_answer}.
  
  On the other hand, expanding $\Tr(PU)^k$ as in \eqref{eq:usingF}, we
  obtain
  \begin{equation}
    \label{eq:M_CUE_expand1}
    \mathcal{M}^{\mathrm{CUE}}_k(N)
    = \sum_{I,J \in \cF_N}
    \left\langle
      U_{i_1, i_\cc1} U_{i_2,i_\cc2}
      \cdots U_{i_k,i_\cc{k}}
      \cc{U}_{j_1, j_\cc1} \cc{U}_{j_2,j_\cc2} \cdots
      \cc{U}_{j_k,j_\cc{k}} 
    \right\rangle_{\mathrm{CUE}(N)},
  \end{equation}
  where the definition of $\cF_N$ remains identical to COE case,
  equation~\eqref{eq:setF_def}.  Averages of products of elements of
  CUE matrices is evaluated using Weingarten functions $\WgU$, see
  Appendix~\ref{sec:summary_CUE}.  The result is usually written using
  two permutations $\sigma$ and $\pi$ acting on indices with and
  without bars correspondingly.  We will view them as two halves of a
  \emph{regular} permutation $\omega$, see
  Definition~\ref{def:RegIrreg}, and will write $\omega = (\sigma,\pi)
  \in S_k \times S_k \subset S_{2k}$.
  Applying~\eqref{eq:Samuel_expansion}, we get
  \begin{align*}
    \mathcal{M}_{k}^{\mathrm{CUE}}(N)
    &=\sum_{\omega=(\sigma,\pi) \in S_{2k}} \WgU (\sigma^{-1}\pi)
    \sum_{I\in \mathcal{F}:I\circ\omega\in \mathcal{F}} 1 \\
    &= \sum_{\omega=(\sigma,\pi)\in S_{2k}}
    \WgU(\sigma^{-1}\pi) \Phi(\omega),
  \end{align*}
  where $\Phi(\omega)$ is defined by \eqref{eq:Phi_definition} and
  therefore obeys Theorem~\ref{thm:countingIJF}.
  
  Since $\WgU(\sigma^{-1}\pi)$ only depends on the cycle structure of
  $\sigma^{-1}\pi$, for the sake of creating the following sum, we
  will write $\WgU(\lambda)$, where $\lambda$ is the cycle structure
  of $\sigma^{-1}\pi$.  It is important to note that for $\omega =
  (\sigma,\pi)$,
  \begin{equation*}
    [\omega, T] = \omega T \omega^{-1} T^{-1}
    = (\sigma\pi^{-1}, \pi\sigma^{-1}),
  \end{equation*}
  and therefore $\lambda$ is half the cycle structure of
  $[\omega, T]$, analogously to the COE case.  Defining
  \begin{equation*}
    S^{\lambda}_{\CUE}:=\{\omega=(\sigma,\pi) \in S_k \times S_k:
    \sigma^{-1}\pi \in S_k
    \text{ has cycle type }\lambda\},
  \end{equation*}
  we may write
  \begin{equation}
    \label{eq:M_CUE_expand2}
    \mathcal{M}_k(N)
    =\sum_{\lambda \vdash k} \WgU(\lambda)
    \sum_{\omega\in S^{\lambda}_{\CUE}}\Phi(\omega)
    =: \sum_{\lambda \vdash k} \WgU(\lambda) P_{\lambda}^R.
  \end{equation}
  From Theorem~\ref{thm:countingIJF} we conclude that $P_{\lambda}^R$
  is a polynomial in $N$.  More precisely, it can be expanded as
  \begin{equation}
    \label{eq:PR_expansion}
    P^R_{\lambda}
    = \left|\Reg{\lambda}{\id}\right| N^k
    + \left|\Reg{\lambda}{2^1}\right| N^{k-1}
    + \left(\left|\Reg{\lambda}{2^2}\right|
      + \left|\Reg{\lambda}{3^1}\right|\right) N^{k-2} + \ldots
  \end{equation}
  Combining equations~\eqref{eq:M_CUE_eval} and
  \eqref{eq:M_CUE_expand2} we get
  \begin{equation}
    \label{eq:M_CUE_expand3}
    k = P^R_{\text{id}}\,\WgU(\id) + P^R_{2^1}\,\WgU(2^1)
    + P^R_{2^2}\,\WgU(2^2) + P^R_{3^1}\,\WgU(3^1) + \ldots,
  \end{equation}
  and no other Weingarten functions have terms of orders $N^{-k}$,
  $N^{-k-1}$ or $N^{-k-2}$.  The relevant terms of these relevant
  functions are (see Appendix~\ref{sec:summary_CUE})
  \begin{align}
    \label{eq:WgU1}
    \WgU(\id) &= \frac{1}{N^k} + \frac{0}{N^{k+1}} + \frac{k(k-1)/2}{N^{k+2}}
                + \ldots \\
    \label{eq:WgU2}
    \WgU(2^1) &= -\frac{1}{N^{k+1}} + \frac{0}{N^{k+2}} + \ldots,\\
    \label{eq:WgU22}
    \WgU(3^1) &= \frac{2}{N^{k+2}} + \ldots,\\
    \label{eq:WgU3}
    \WgU(2^2) &= \frac{1}{N^{k+2}} + \ldots
  \end{align}
  Substituting these expansions together with \eqref{eq:PR_expansion}
  into~\eqref{eq:M_CUE_expand3} and collecting the terms contributing
  to the first three orders of a $1/N$ expansion, we obtain
  \begin{align*}
    k = \left|\Reg{\id}{\id}\right|
    &+ \left(\left|\Reg{\id}{2^1}\right|
    - \left|\Reg{2^1}{\id}\right|\right)\frac1N \\
    &+ \left(\frac{k(k-1)}{2}\left|\Reg{\id}{\id}\right|
    + \left|\Reg{\id}{2^2}\right|
    + \left|\Reg{\id}{3^1}\right|
    - \left|\Reg{2^1}{2^1}\right|
    + 2 \left|\Reg{3^1}{\id}\right|
    + \left|\Reg{2^2}{\id}\right|
    \right)\frac1{N^2} + \ldots
  \end{align*}
  We already know that $\left|\Reg{\id}{\id}\right|=k$ from \eqref{eq:alphas}
  and the $\frac1N$ term we get $\left|\Reg{2^1}{\id}\right| =
  \left|\Reg{\id}{2^1}\right| = 0$.  Finally, from the $\frac1{N^2}$
  term we get
  \begin{equation}
    \label{eqn:regular third term}
    \frac{k^2(k-1)}{2} + \left|\Reg{\id}{2^2}\right|
    +\left|\Reg{\id}{3^1}\right| - \left|\Reg{2^1}{2^1}\right|
    + 2 \left|\Reg{3^1}{\id}\right|
    + \left|\Reg{2^2}{\id}\right| = 0.
  \end{equation}

  Now we similarly expand our primary target, the moment function
  $M_k(N)$ for COE.  Combining equation~\eqref{eq:kmoment_expand} with
  Theorem~\ref{thm:countingIJF} and Definition~\ref{def:RegIrreg}, and
  using the fact that $\Reg{\id}{2^1}$, $\Reg{2^1}{\id}$ and
  $\Irreg{\alpha}{\id}$ are empty to reduce the number of terms, we
  get
  \begin{align}
    M_k(N) = {}
    &\WgO(\id) \left(\left|\Reg{\id}{\id}\right| N^k
      + \left|\Irreg{\id}{2^1}\right| N^{k-1}
      + \left( \left|\Reg{\id}{2^2}\right|
      + \left|\Irreg{\id}{2^2}\right|
      + \left|\Reg{\id}{3^1}\right|
      + \left|\Irreg{\id}{3^1}\right| \right) N^{k-2} \right)
      \nonumber \\
    &+ \WgO(2^1) \left(\left|\Reg{2^1}{2^1}\right|
      + \left|\Irreg{2^1}{2^1}\right| \right) N^{k-1}
      \nonumber \\
    \label{eq:momentExpandSets}
    &+ \WgO(2^2) \left|\Reg{2^2}{\id}\right| N^k
    + \WgO(3^1) \left|\Reg{3^1}{\id}\right| N^k + \mathcal{O}\left(\frac{1}{N^3}\right).
  \end{align}

  The necessary Weingarten functions are (see Appendix~\ref{sec:summary_COE})
  \begin{align}
    \label{eq:WgO1}
    \WgO(\id) &= \frac{1}{N^k} - \frac{k}{N^{k+1}} +
                \frac{k(3k-1)/2}{N^{k+2}} + \ldots \\
    \label{eq:WgO2}
    \WgO(2^1) &= -\frac{1}{N^{k+1}} + \frac{k+2}{N^{k+2}} + \ldots,\\
    \label{eq:WgO3}
    \WgO(3^1) &= \frac{2}{N^{k+2}} + \ldots,\\
    \label{eq:WgO22}
    \WgO(2^2) &= \frac{1}{N^{k+2}} + \ldots
  \end{align}
  Substituting these together with irregular counts given in
  Lemma~\ref{lem:irregular_perm}, we get
  \begin{align*}
    M_k(N) =
    & \left(1 - \frac{k}{N} + \frac{3k^2-k}{2N^2} \right)
      \left(k + \frac{k^2}{N} 
      + \frac{k^3-2k^2 + \left|\Reg{\id}{2^2}\right|
      + \left|\Reg{\id}{3^1}\right|}{N^2} \right)
      \nonumber \\
    & - \frac1{N^2}\left(\left|\Reg{2^1}{2^1}\right| + k^2(k-2)\right) 
      + \frac1{N^2}\left|\Reg{2^2}{\id}\right|
      + \frac2{N^2}\left|\Reg{3^1}{\id}\right|
      + \mathcal{O}\left(\frac1{N^3}\right).
  \end{align*}
  Expanding and collecting terms, we use \eqref{eqn:regular third
    term} to get
 
  \begin{align*}
    M_k(N) = k + \frac0N + \frac{0}{N^2}
    + \mathcal{O}\left(\frac1{N^3}\right)
    = k + \mathcal{O}\left(\frac1{N^3}\right),
  \end{align*}
  which is the desired result.
\end{proof}

\section*{Acknowledgment}

This material is based upon work supported by the National Science
Foundation under Grant No.~1815075 and by the Binational Science
Foundation under Grant No.~2016281.  The idea for this research
project arose during discussions with Chris Joyner and we are
profoundly grateful for his help and support.  We would like to thank
Jon Keating and the interest he took in our results and for bringing
the work \cite{Mec_tams08} to our attention.  We are grateful to the
anonymous referee for several improving suggestions.

\section*{Data Availability}

Data sharing is not applicable to this article as no new data were
created or analyzed in this study.

\appendix
\section{Weingarten calculus}
\label{sec:Weing_calc}

\subsection{Circular Unitary Ensemble \texorpdfstring{(the unitary group $U(N)$)}{}}
\label{sec:summary_CUE}

Unitary Weingarten functions are building blocks for integration of
products of matrix elements over the unitary group.  Let $k \leq N$
and let $\mathbf{i} = (i_1, \ldots, i_k)\in [N]^k$ and
$\cc{\mathbf{i}} = (i_\cc1, \ldots, i_\cc{k})\in [N]^k$ be two
arbitrary sequences of indices from $[N] := \{1,\ldots,N\}$ with
\emph{distinct} entries.  The defining property of the Weingarten
function of $\sigma \in S_k$ is
\begin{equation}
  \label{eq:WgU_def}
  \WgU(\sigma) = \int_{U(N)}
  U_{i_1 i_\cc1} U_{i_2 i_\cc2} \cdots U_{i_k i_\cc{k}}
  \cc{U}_{i_{\sigma(1)} i_\cc1} \cc{U}_{i_{\sigma(2)} i_\cc2}
  \cdots \cc{U}_{i_{\sigma(k)} i_\cc{k}} \,du,
\end{equation}
where $du$ is the uniform (Haar) measure on the unitary group $U(N)$.
Due to invariance properties of the measure, the value of the function
is independent of the choice of the sequences $\mathbf{i}$ and
$\cc{\mathbf{i}}$; both are often taken to be be equal to
$(1,2,\ldots,k)$, although that tends to obscure the meaning of
$\sigma$ as a permutation acting on the order of indices $i$ rather
than their values.  Furthermore, Weingarten functions depend only on
the equivalence class of $\sigma$, which is in turn determined by its
cycle structure.

Knowledge of Weingarten functions is enough to evaluate the average of
a general product of the elements of a $U\in U(N)$ via the following
formula.    Let $\mathbf{i}, \cc{\mathbf{i}} \in [N]^k$ and $\mathbf{j},
  \cc{\mathbf{j}} \in [N]^{k'}$.  Then
\begin{equation}
  \label{eq:Samuel_expansion}
  \int_{U(N)}
  U_{i_1 i_\cc1} 
  \cdots U_{i_k i_\cc{k}}
  \cc{U}_{j_1 j_\cc1} 
  \cdots \cc{U}_{j_{k'} j_{\cc{k}'}} \,du
  = \delta_{k,k'} \sum_{\substack{\sigma \in S_k \colon
    \sigma(\mathbf{i}) = \mathbf{j} \\
    \pi \in S_k \colon
    \pi(\cc{\mathbf{i}}) = \cc{\mathbf{j}}}}
  \WgU(\sigma^{-1}\pi),
\end{equation}
where $\delta_{k,k'}$ is the Kronecker delta function.

One of the possible way to compute Weingarten functions is via the
orthogonality relations
\begin{equation}
  \label{eq:recursion_CUE}
  N\WgU(\omega) + \sum_{i=1}^{k-1}\WgU\big((i\,k)\omega\big)
  = \delta_{\omega(k),k}
  \WgUk{k-1}\left(\omega^{\downarrow}\right),
\end{equation}
where $(i\,k)$ is the transposition between $i$ and $k$, and
where $\omega^{\downarrow}$ is the restriction of $\omega$ from $S_k$
to $S_{k-1}$, which is well defined here due to condition
$\omega(k)=k$ enforced by the Kronecker delta.
Relations~\eqref{eq:recursion_CUE} together with the initial condition
$\WgUk{0}(\emptyset) = 1$ (or, one step up, $\WgUk{1}\big((1)\big)=1/N$ fully
determine the Weingarten functions for all $N$ and $k \leq N$.

We will denote the cycle structure of $\sigma$ as
$1^{\alpha_1}2^{\alpha_2}...k^{\alpha_k}$, where $\alpha_j$ is the
number of cycles of $\sigma$ of length $j$.  We note that
$\alpha_1+\alpha_2+...+\alpha_k = \ell(\sigma)$ and
$\alpha_1+2\alpha_2+...+k\alpha_k=k$.  Since a Weingarten function
depends only on the cycle structure of $\sigma$, we will abuse notation
slightly and write $\WgU(1^{\alpha_1}2^{\alpha_2}...k^{\alpha_k})$ for
$\WgU(\sigma)$ when convenient.

The first term in the asymptotic expansion of a Weingarten function as
$N\to\infty$ is
\begin{equation}
  \label{eqn:CUEproduct}
  \WgU(1^{\alpha_1}2^{\alpha_2}...k^{\alpha_k})
  = \prod_{j=1}^k \left(\mathrm{Wg}^{\mathrm{CUE}}_{N,j}(j^1)\right)^{\alpha_j}
  + \mathcal{O}(N^{\ell-2k-2}),
  \qquad \ell = \ell(\sigma),
\end{equation}
where 
\begin{equation}
  \label{eqn:CUEpart}
  \mathrm{Wg}^{\mathrm{CUE}}_{N,j}(j^1)
  = (-1)^{j-1}  C_{j-1} N^{1-2j}
  + \mathcal{O}(N^{-1-2j}), \qquad C_{j-1} =\frac{1}{j} \binom{2j-2}{j-1},
\end{equation}
where $C_n$ are the Catalan numbers.
We note that the product $\prod_{j=1}^k \left(\mathrm{Wg}^{\mathrm{CUE}}_{N,j}(j^1)\right)^{\alpha_j}$
is of order $N^{\ell-2k}$, meaning the asymptotic expansion of the CUE
Weingarten function does not have a term of the order
$N^{\ell-2k-1}$.  It should be emphasized that $j^1$ refers to cycle
structure of a permutation with $j$ elements.

At this point we provide a historical information on Weingarten
calculus, to the best of our knowledge.  Weingarten functions were
first defined and systematically studied by Samuel \cite{Sam_jmp80} who
obtained expansion~\eqref{eq:Samuel_expansion}, orthogonality
relations \eqref{eq:recursion_CUE} as well as an expression for $\WgU$
in terms of characters of $S_k$ (the proof of the expression is
attributed in \cite{Sam_jmp80} to Fritz Beukers).  The function is
named after Weingarten who in an earlier work \cite{Wei_jmp78}
obtained asymptotic results equivalent to \eqref{eqn:CUEproduct}.
Averages over unitary group were used extensively in physics (see
\cite{BroBee_jmp96} for one of \emph{many} applications, to quantum
transport) and were eventually rediscovered in the mathematical
literature by Collins \cite{Col_imrn03}.  A beautiful interpretation
of asymptotic coefficients of the Weingarten function as the number of
monotone factorizations by Matsumoto and Novak \cite{MatNov_imrn13}
allowed one of the present authors with Kuipers \cite{BerKui_jmp13a}
to put the use of random matrix theory in quantum chaotic transport on
a more solid mathematical basis.  Notation in the present section is
kept in line with the mathematical sources such as
\cite{Mat_rmta13,ColMat_alea17}.

In the present work we use the first few terms of Weingarten functions
for specific cycle structures, whose calculation we will now discuss.
More precisely, we need all Weingarten functions whose expansion has
terms of order $N^{-k}$, $N^{-k-1}$, and $N^{-k-2}$.  From
\eqref{eqn:CUEproduct}-\eqref{eqn:CUEpart}, these Weingarten functions
are $\WgU(\id)$, $\WgU(2^1)$, $\WgU(3^1)$ and $\WgU(2^2)$.  In our
notation we have omitted $1^{\alpha_1}$ and any factor where $\alpha_j=0$, for
the sake of brevity.  The leading order term of each function can be
obtained from the product in \eqref{eqn:CUEproduct}, while the next order
term is 0 because the error bound in \eqref{eqn:CUEproduct}.  The
resulting expansions are given by \eqref{eq:WgU1}-\eqref{eq:WgU3}, but
we have yet to determine the third term in the expansion of 
\begin{equation}
  \label{eq:WgU1_unknown}
  \WgU(\id) = \frac{1}{N^k} + \frac{0}{N^{k+1}}
  + \frac{t_k}{N^{k+2}} + \ldots   
\end{equation}

In order to do so, we write out equation \eqref{eq:recursion_CUE} as
it applies to $\id\in S_k$.
\begin{equation}
  \label{eq:recursion_id}
  N\WgU(\id) + (k-1)\WgU(2^1)
  = \mathrm{Wg}^{\mathrm{CUE}}_{N,k-1}\left(\id \right),
\end{equation}
or, extracting the coefficients of the term $1/N^{k+1}$ on both sides,
\begin{equation}
  \label{eq:recursion_t}
  t_k = k-1 + t_{k-1}.
\end{equation}
Since $\WgUk{1}(\id_1) = 1/N$, we have $t_1=0$ and therefore $t_k =
k(k-1)/2$.  We conclude that
\begin{equation*}
  \WgU(\id) = \frac{1}{N^k} + \frac{0}{N^{k+1}}
  + \frac{k(k-1)/2}{N^{k+2}} + \ldots
\end{equation*}

We remark that the coefficient $k(k-1)/2$ can also be obtained as the
number of primitive (monotone) factorizations of the identity into two
transpositions \cite{MatNov_imrn13}.

\subsection{Circular Orthogonal Ensemble \texorpdfstring{(compact symmetric space $U(N)/O(N)$)}{}}
\label{sec:summary_COE}

Circular Orthogonal Ensemble was introduced in the seminal article
\cite{Dys_jmp62_i} (see also its ``prequel'' published later the same
year \cite{Dys_jmp62_0}).  The rest of the references is given
at the end of this section, after the relevant facts are stated.

For a given $\omega \in S_{2k}$, considered as a permutation of
$\setK$, the COE Weingarten function is defined by
\begin{equation}
  \label{eq:WgO_def}
  \WgO(\omega) = \int_{\COE(N)}
  U_{i_1 i_\cc1} U_{i_2 i_\cc2} \cdots U_{i_k i_\cc{k}}
  \cc{U}_{i_{\omega(1)} i_{\omega(\cc1)}}
  \cc{U}_{i_{\omega(2)} i_{\omega(\cc2)}}
  \cdots \cc{U}_{i_{\omega(k)} i_{\omega(\cc{k})}} \,du,
\end{equation}
where $du$ is the COE probability measure and
$(i_1, i_\cc1, \ldots, i_k, i_\cc{k}) =: I$ is an arbitrary
sequence of \emph{distinct} values from $[N]$.  The average of any
product of matrix elements can now be calculated as follows: for any
$I\colon \setK \to [N]$ and $J \colon \setK'\to[N]$
\begin{equation}
  \label{eq:COE_expansion}
  \int_{\COE(N)}
  U_{i_1 i_\cc1} 
  \cdots U_{i_k i_\cc{k}}
  \cc{U}_{j_1 j_\cc1} 
  \cdots \cc{U}_{j_{k'} j_\cc{k'}} \,du
  = \delta_{k,k'} \sum_{\omega \in S_{2k} \colon
    I\circ\omega = J}
  \WgO(\omega).
\end{equation}
Orthogonality relations for $\WgO$ take the following form:
\begin{multline}
  \label{eq:recursion_COE}
  (N+1)\WgO(\omega)
  + \sum_{z=1}^{k-1} \WgO\big((z\,k)\omega\big)
  + \sum_{z=1}^{k-1} \WgO\big((\cc{z}\,k)\omega\big) \\
  = \delta_{\{\omega(k),\omega(\cc{k})\},\{k,\cc{k}\}}
  \mathrm{Wg}^{\mathrm{COE}}_{N,k-1}(\omega^{\downarrow}),
\end{multline}
where the right-hand side is non-zero if and only if $\omega$ leaves
the set $\{k,\cc{k}\}$ invariant, in which case $\omega^\downarrow$ is
the natural projection of $\omega$ to $S_{2(k-1)}$.

From invariance properties of COE one deduces that $\WgO$ depends only
on the cycle structure of the permutation
$[\omega,T]=\omega T \omega^{-1} T^{-1}$.  The latter permutation has
an even number of cycles of any length.  We will therefore use
$\WgO(1^{\alpha_1} 2^{\alpha_2} \cdots)$ to denote the Weingarten
function of a permutation $\omega$ such that $[\omega,T]$ has
$2\alpha_1$ cycles of length 1, $2\alpha_2$ cycles of length $2$ and
so on.  To give an example, the cycle structure corresponding to
$\omega=(1,\cc{b})$ is $1^k$ if $b=1$ and $1^{k-2}2^1$ if $b\neq 1$.
As before, we have $\alpha_1+2\alpha_2+...+k\alpha_k=k$ and we let
$h = \frac12\ell[\omega,T] := \alpha_1+\alpha_2+...+\alpha_k$.  We also
note that since $T$ is an involution, $[\omega^{-1},T]$ has the same
cycle structure as $[\omega,T]$ and in some sources formulas such as
\eqref{eq:COE_expansion} are written in terms of $\omega^{-1}$
instead.

For a fixed $k$, a Weingarten function is a rational function of $N$
with the asymptotic expansion in $1/N$ given by \cite[Sec.~IV]{BroBee_jmp96}
\begin{equation}
  \label{eqn:COEproduct}
  \WgO(1^{\alpha_1}2^{\alpha_2}...k^{\alpha_k})
  = \prod_{j=1}^k \left(\mathrm{Wg}^{\COE}_{N,j}(j^1)\right)^{\alpha_j}
  + \mathcal{O}(N^{h-2k-2}),
\end{equation}
with
\begin{equation}
  \label{eqn:COEpart}
  \mathrm{Wg}^{\COE}_{N,j}(j^1)
  = (-1)^{j-1} C_{j-1} N^{1-2j}
  - (-4)^{j-1} N^{-2j} + \mathcal{O}(N^{-1-2j}),
  \quad C_{j-1} = \frac{1}{j} \binom{2j-2}{j-1}.
\end{equation}
We note that unlike \eqref{eqn:CUEproduct},
expansion~\eqref{eqn:COEproduct} does contain terms of order
$N^{h-2k-1}$ but they can be obtained from expanding the first
product.

We will need all functions $\WgO$ which have terms of order $N^{-k-2}$
or above.  According to \eqref{eqn:COEproduct}-\eqref{eqn:COEpart}
those are $\WgO(\id)$, $\WgO(2^1)$, $\WgO(2^2)$ and $\WgO(3^1)$.
The first two terms of each expansion can be obtained from the product in
\eqref{eqn:COEproduct}.  For the factors one can either use
\eqref{eqn:COEpart} or explicit formulas \cite[Table~II and
IV]{BroBee_jmp96},
\begin{align}
  \label{eq:WgO1explicit}
  \mathrm{Wg}^{\COE}_{N,1}(1^1) &= \frac1{N+1},\\
  \nonumber
  \mathrm{Wg}^{\COE}_{N,2}(2^1) &= \frac{-1}{N(N+1)(N+3)}, \\
  \nonumber
  \mathrm{Wg}^{\COE}_{N,3}(3^1) &= \frac{2}{(N-1)N(N+1)(N+3)(N+5)}.  
\end{align}
The results are given in \eqref{eq:WgO1}-\eqref{eq:WgO22}.

We now employ \eqref{eq:recursion_COE} with $\omega = \id$ to
determine the last term in the expansion
\begin{equation}
  \label{eq:WgO1_unknown}
  \WgO(\id) = \frac{1}{N^k} - \frac{k}{N^{k+1}}
  + \frac{t_k}{N^{k+2}} + \ldots
\end{equation}
Since $\omega' = (z\, k)$ gives rise to cycle
structure $2^1$ for any $z \neq k, \cc{k}$, we have
\begin{equation}
  \label{eq:WgOid}
  (N+1)\WgO(\id) + 2(k-1)\WgO(2^1)
  = \mathrm{Wg}^{\COE}_{N,k-1}(\id).
\end{equation}
Substituting expansions~\eqref{eq:WgO1} and \eqref{eq:WgO2} and
extracting the coefficient of $N^{-k-1}$ we get
\begin{equation*}
  t_k=t_{k-1} + 3k-2.
\end{equation*}
From \eqref{eq:WgO1explicit} we have $t_1=1$ resulting in
\begin{equation}
  \label{eq:tkO_answer}
  t_k = \sum_{j=1}^k (3j - 2) = \frac{3k^2-k}{2}.
\end{equation}

Much of the material of this section has first appeared in
\cite{BroBee_jmp96} albeit without derivation.  Careful mathematical
treatment was done more recently by Matsumoto and co-authors
\cite{Mat_rmta12,Mat_rmta13,ColMat_alea17}, whose notation we mostly
follow.  In particular, relation \eqref{eq:recursion_COE} first
appears in \cite[Eq.~(4.3)]{BroBee_jmp96} in terms of a different
notation (using the list of cycle lengths in $[\omega,T]$) but an
elegant proof is given in \cite[Lemma 5.3]{ColMat_alea17}.

Some other possibly relevant results that we do not use here are as
follows.  The coefficients of the asymptotic expansion of COE
Weingarten function as $N\to\infty$ were shown \cite[Proof of
Thm.~6.4]{BerKui_jmp13a} to count palindromic monotone (or primitive)
factorizations, by analogy with the monotone factorizations for CUE
\cite{MatNov_imrn13}.  This can be used for alternative derivation of
\eqref{eq:WgO1}-\eqref{eq:WgO22} together with \eqref{eq:tkO_answer}.  
COE Weingarten functions can be calculated as pseudo-inverse of power functions
\cite[Eq.~(2.5)]{Mat_rmta13} through their simple relation with
Weingarten function for the orthogonal group $O(N)$, see
\cite{Mat_rmta12} and \cite[Thm.~5.4]{ColMat_alea17}.  Uniform bound
in terms of $k$ obtained in \cite{ColMat_alea17} could be crucial to
considering the simultaneous limit $k,N \to \infty$ and accessing the
spectral form factor.

\section{Classifying regular permutations}
\label{sec:counting_regular}

In this section we prove Lemma~\ref{lem:regular_perm}, in several
steps.  It is also a warm up for the more involved proof of
Lemma~\ref{lem:irregular_perm} reported in
Appendix~\ref{sec:class_irreg}.

\begin{lemma}
\label{lem:T id}
If $[\omega,T]=\id$, then for $z_1,z_2 \in \mathcal{K}$, $\omega(z_1)=z_2$ if and only if $\omega(\overline{z_1})=\overline{z_2}$.
\end{lemma}

\begin{proof}
  Let $z_1,z_2 \in \mathcal{K}$ such that $\omega(z_1)=z_2$. Note that
  $\overline{\overline{z_1}}=z_1$. Observe that since
  $[\omega,T]=\id$,
  \begin{equation}
    \overline{z_2} = [ \omega,T] (\overline{z_2})
    = \omega T\omega^{-1} T^{-1}(\overline{z_2})
    = \omega T\omega^{-1}(z_2)
    = \omega T(z_1) = \omega(\overline{z_1}).
  \end{equation}
  Hence $\omega(\overline{z_1})=\overline{z_2}$. The other direction
  comes immediately from replacing $\overline{z_2}$ with $z_2$ in the
  calculation above.
\end{proof}

\begin{lemma}
  \label{lem:Q id}
  If $[\omega,Q]=\id$, then for $m,p\in [k]$,
  $\omega(\overline{m})=\overline{p}$ if and only if
  $\omega(m+1)=p+1$.
\end{lemma}

\begin{proof}
  Let $m,p\in [k]$ such that
  $\omega(\overline{m})=\overline{p}$. Observe that since
  $[\omega,Q]=\id$,
  \begin{equation}
    p+1=[\omega,Q](p+1)
    =\omega Q\omega^{-1} Q^{-1}(p+1)
    =\omega Q \omega^{-1} (\overline{p})
    =\omega Q(\overline{m})
    =\omega(m+1).
  \end{equation}
  Hence $\omega(m+1)=p+1$. The reverse direction is calculated similarly.
\end{proof}

\begin{lemma}
  $\Reg{\id}{\id}= \{ s^n:0\leq n\leq k-1  \}$.
\end{lemma}

\begin{proof}
  Since all regular permutations $\omega$ map $[k]$ to $[k]$ and
  $[\overline{k}]$ to $[\overline{k}]$, we know for some $m,p\in [k]$,
  $\omega(m)=p$. Lemma \ref{lem:T id} implies that
  $\omega(\overline{m})=\overline{p}$, which by Lemma \ref{lem:Q id}
  implies $\omega(m+1)=p+1$, etc. This means that $\omega=s^{p-m}$.
\end{proof}

\begin{lemma}
  \label{lem:reg12}
  $\Reg{\id}{2^1} = \emptyset$.
\end{lemma}

\begin{proof}
  Suppose there is an $\omega \in \Reg{\id}{2^1}$, which by the
  definition must contain exactly one directed cycle of length 4.
  Denote the solid edges involved in this cycle by $\cc{l} \to l+1$
  and $\cc{m-1} \to m$.  Since $\omega$ is regular, and the dashed
  directed edges in the graph must go from an element of $[\cc{k}]$ to
  an element of $[k]$, the dashed edges involved in the cycle of
  length 4 must be $\cc{l}\to m$ and $\cc{m-1} \to l+1$.  All other
  cycles, both directed and undirected, have length 2 and thus every
  other dashed directed (corresp.\ undirected) edge coincides with a
  sold directed (corresp.\ undirected) edge.  The graph therefore fits
  the shape in Figure~\ref{fig:square directed four-cycle} which is
  not permissible, since it is not possible for the dashed edges to
  form a single cycle.  Thus the set $\Reg{\id}{2^1}$ is empty.

  \begin{figure}[ht]
    \centering
    \includegraphics[scale=0.5]{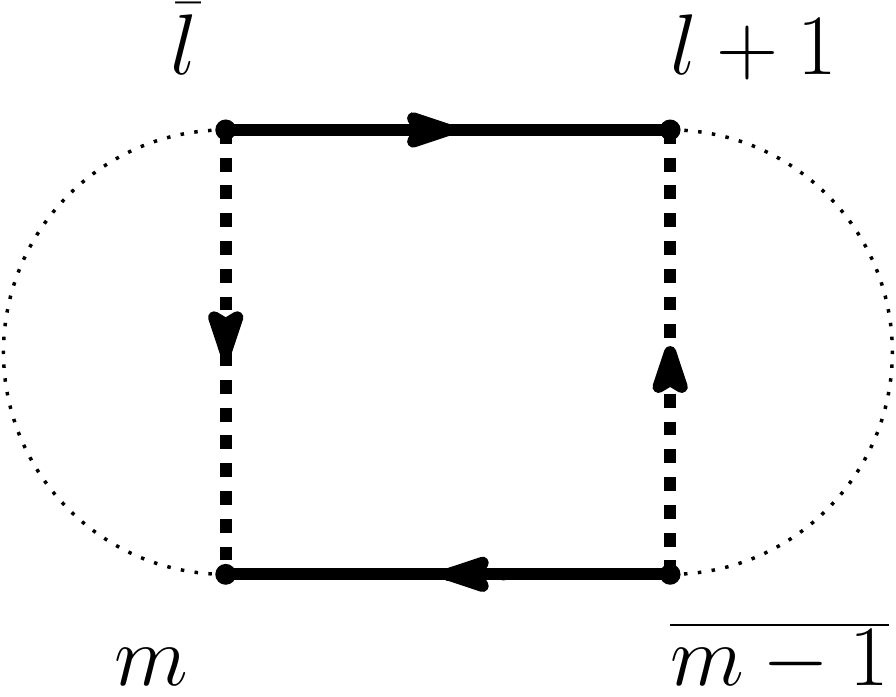}
    \caption{A candidate graph of $\omega \in
      \Reg{\id}{2^1}$.}
    \label{fig:square directed four-cycle}
  \end{figure}

  For comparison, we also provide an algebraic proof.  Notice that for
  all $\omega \in S_{2k}$, $\omega Q T \omega^{-1}\in S_{2k}$ consists
  of two cycles of length $k$, namely

  \begin{equation}
    \label{eq:omega_cycle}
    \big(\omega(1) \; \omega(2)\; \ldots \; \omega(k)\big)
    \big(\omega(\cc{k}) \; \ldots\; \omega(\cc2) \; \omega(\cc1) \big).
  \end{equation}
  Since $\alpha= \id$, $\omega T=T \omega$. Since $\omega$ is regular,
  $\omega Q \omega^{-1}Q^{-1}$ leaves $[\cc k]$ invariant. We know
  that $\omega Q \omega^{-1}Q^{-1}$ has only one nontrivial cycle in
  $[\cc k]$ and that this cycle has length 2. Let
  $(\cc q_1\; \cc q_2)$ be this cycle. Then we have
  \begin{align*}
    \omega Q (\cc q_1)
    &= Q \omega (\cc q_2),
      \qquad
      \omega Q (\cc q_2) = Q \omega (\cc q_1), \\
    \omega Q( \cc b)
    &= Q \omega (\cc b) \quad
      \mbox{for all } \cc b \in [\cc k] \setminus \{\cc q_1, \cc q_2\}.
  \end{align*}

  Introduce the notation $r_1 = \omega(q_1) = \omega T(\cc q_1)$ and
  $r_2=\omega(q_2) = \omega T(\cc q_2)$ and assume, without loss of
  generality, that $r_1 < r_2$.  We will next understand the action of
  $\omega Q T \omega^{-1}$ on $r \neq r_1,r_2$.  We have
  $T \omega^{-1}(r)\neq \cc q_1, \cc q_2$ and therefore
  \begin{align}
    \omega Q T \omega^{-1}(r)
    = Q \omega T \omega^{-1}(r)
    = QT \omega \omega^{-1}(r)
    = QT(r)
    = r+1.
  \end{align}
  On the other hand,
  \begin{align}
    \omega QT \omega^{-1}(r_2)
    = \omega Q (\cc q_2)
    = Q \omega (\cc q_1)
    = Q \omega T (q_1)
    = QT \omega (q_1)
    = QT(r_1)
    = r_1+1.
  \end{align}
  Then repeated compositions of $\omega Q T \omega^{-1}$ will produce the cycle:
  \begin{equation}
    \big(r_2 \; (r_1+1) \; (r_1+2) \; \ldots \; (r_2-1) \big).
  \end{equation}
  We have just obtained a cycle of $\omega QT \omega^{-1}$ that does
  not contain $r_1$, but according to~\eqref{eq:omega_cycle}, a single
  cycle must contain all elements of $[k]$. This is a contradiction
  and therefore $\Reg{\id}{2^1} = \emptyset$.
\end{proof}

\begin{remark}
  The algebraic proof of Lemma~\ref{lem:reg12} has certain advantages
  (it is easier to check that every eventuality is considered), but it
  is certainly longer and harder to construct.  In fact, it was
  constructed in the first place by mapping the ``graph-based'' proof
  into algebraic properties: for example, the ``dashed edges form a
  single cycle'' corresponds to the cycle structure of
  $\omega Q T \omega^{-1}$, equation~\eqref{eq:omega_cycle}.

  For this reason, we will use ``graph-based'' arguments for the even
  more sophisticated proofs of Appendix~\ref{sec:class_irreg}.
\end{remark}

\section{Classifying irregular permutations}
\label{sec:class_irreg}

In this section we will prove the assertions made in Lemma
\ref{lem:irregular_perm} . In the computations in this section, note
that all addition and subtraction performed on elements of
$\mathcal{K}:=[k]\cup [\overline{k}]$ is done modulo $k$, even though
not explicitly stated each time.

We will be classifying the irregular permutations into the sets
$\Irreg{\alpha}{\beta}$ based on the cycle structures of $[\omega,T]$
and $[\omega,Q]$, see Definition~\ref{def:RegIrreg}.  Since the cycle
structures of these commutators are determined by the types of cycles in
$G_{\omega}$, if $G_{\omega}$ and $G_{\omega'}$ are
isomorphic, then $\omega$ and $\omega'$ are in the same set
$\Irreg{\alpha}{\beta}$.  Thus we may further dissect our problem into
finding the representative permutation $\omega$ corresponding to each isomorphism class of
graphs coming from $\Irreg{\alpha}{\beta}$, and expressing the rest of
them as $s^n \omega s^m$ for $n,m\in [k]$ in accordance with
Lemma~\ref{lem:counting_cycles} part (\ref{item:graphs isomorphic}).

Our task is to classify graphs with the majority of
directed and undirected cycles having length two; those are formed by
dashed edges coinciding with solid edges.  The direction must also
coincide (if the edges are directed) to make the cycle balanced.
Dashed edges not coinciding with a solid edge will be called
\emph{chords} (see Figure~\ref{fig:balanced_cycles} for the visual
reason behind the term \emph{chord}).  To put it another way, chords
are the dashed edges belonging to a cycle of length greater than 2.

We will also frequently appeal to another property of $G_\omega$
discussed in Section~\ref{sec:graph_model}: the dashed subgraph forms
a single cycle with all directed edges pointing in the same direction
along the cycle.

\begin{lemma}
  \label{lem:irreg at least one chord}
  $\Irreg{\alpha}{\id}=\emptyset$ for any $\alpha$.
\end{lemma}

\begin{proof}
  Since $\beta=\id$ we have no directed chords.  An irregular
  permutation $\omega$, by definition, satisfies
  $\omega(q)=\overline{l}$ for some $q\in [k]$,
  $\overline{l}\in [\overline{k}]$.  Consider the directed dashed edge
  $(\omega(\overline{q-1}) \to \omega(q)) = (\omega(\overline{q-1})
  \to \overline{l})$. In order to not create a chord, it would have to
  match an existing solid directed edge, but the only solid directed
  edge incident to $\overline{l}$ is $(\overline{l}\to l+1)$, which
  points in the opposite direction. This would make an unbalanced
  2-cycle, so it is not allowed. Hence, when considering irregular
  permutations, we will never have $[\omega,Q]=\id$.
\end{proof}

\begin{lemma}
  \label{lem:irreg id and 2^1}
  $\Irreg{\id}{2^1}=\{ s^n(1 \hspace{1mm} \overline{1})s^m:n,m\in [k]
  \}$, $|\Irreg{\id}{2^1}|=k^2$.
\end{lemma}

\begin{proof}
  If $\omega \in \Irreg{\id}{2^1}$, the graph model $G_\omega$ has
  exactly one directed 4-cycle and no undirected cycles larger than
  2. We know from the proof of Lemma~\ref{lem:irreg at least one
    chord} that we have at least the directed chord
  $(\omega(\overline{q-1}) \to \overline{l})$. Since we have only one
  directed 4-cycle, this chord must be part of it. Certainly
  $\omega(\overline{q-1})$ must equal either $m$ or $\overline{m-1}$
  for some $m\in [k]$. We find the former case cannot work because we
  would have have 3 arrows pointing in the same direction in a cycle
  of length 4, making a balanced cycle impossible, as shown in Figure
  \ref{fig:irreg directed 4-cycle shape wrong}.  For the latter case,
  the diagram looks like Figure \ref{fig:irreg directed 4-cycle shape
    right}, an arrangement which can be completed to become a balanced
  cycle.

  \begin{figure}
    \centering
    \begin{subfigure}[b]{0.4\textwidth}
      \centering
      \includegraphics[scale=0.5]{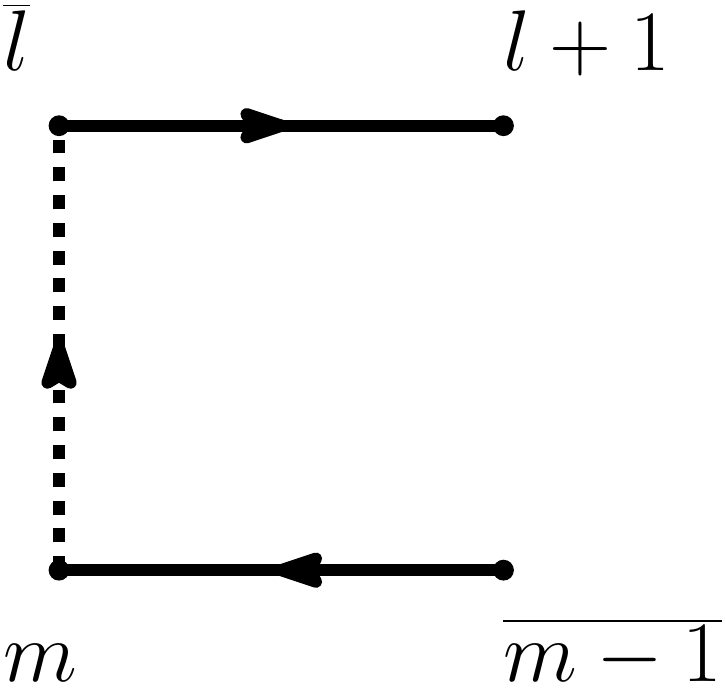}
      \caption{balanced 4-cycle cannot be formed}
      \label{fig:irreg directed 4-cycle shape wrong}
    \end{subfigure}
    \begin{subfigure}[b]{0.4\textwidth}
      \centering
      \includegraphics[scale=0.5]{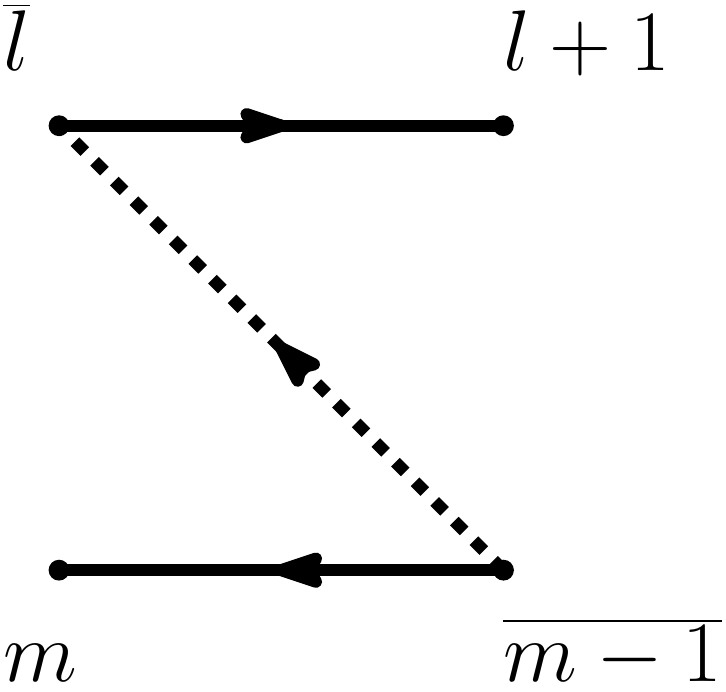}
      \caption{balanced 4-cycle can be formed}
      \label{fig:irreg directed 4-cycle shape right}
    \end{subfigure}
    \caption{Two possibilities for the chord $(\omega(\overline{q-1})
      \to \overline{l})$.}
  \end{figure}

  The unique way to complete Figure~\ref{fig:irreg directed 4-cycle
    shape right} to produce a balanced 4-cycle is shown in
  Figure~\ref{fig:directed four-cycle}.  Note that we exclude the case
  $m=l+1$ since it produces loops and makes no sense in our graph
  model.
  
  \begin{figure}
    \centering
    \includegraphics[scale=0.5]{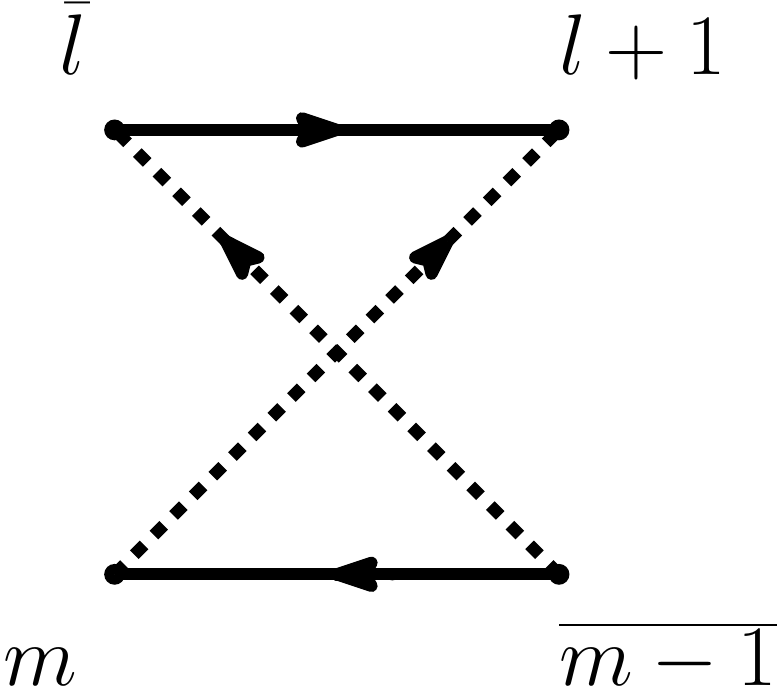}
    \caption{A valid 4-cycle with the chord $(m \to \cc{l})$.}
    \label{fig:directed four-cycle}
  \end{figure}

  \begin{figure}
    \centering
    \begin{subfigure}[b]{0.45\textwidth}
      \centering
      \includegraphics[scale=0.5]{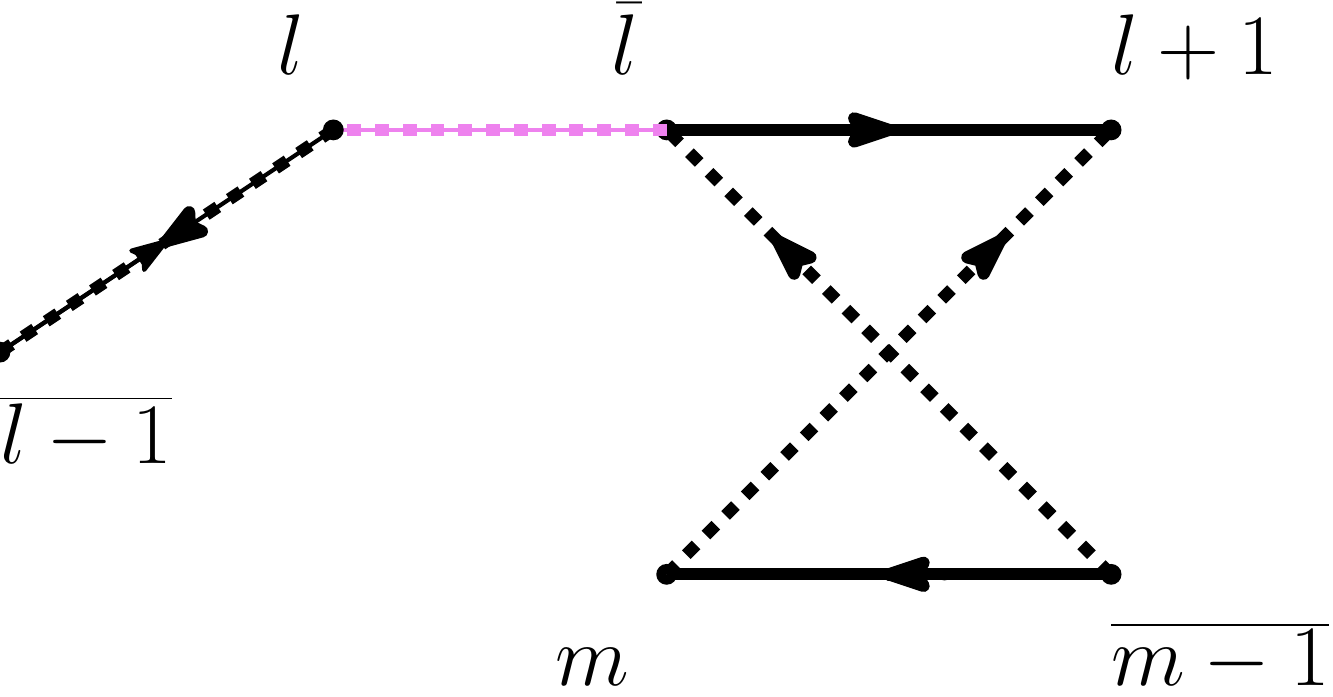}
      \caption{arc completion is not possible}
      \label{fig:directed four-cycle with bad path}
    \end{subfigure}
    \begin{subfigure}[b]{0.45\textwidth}
      \centering
      \includegraphics[scale=0.5]{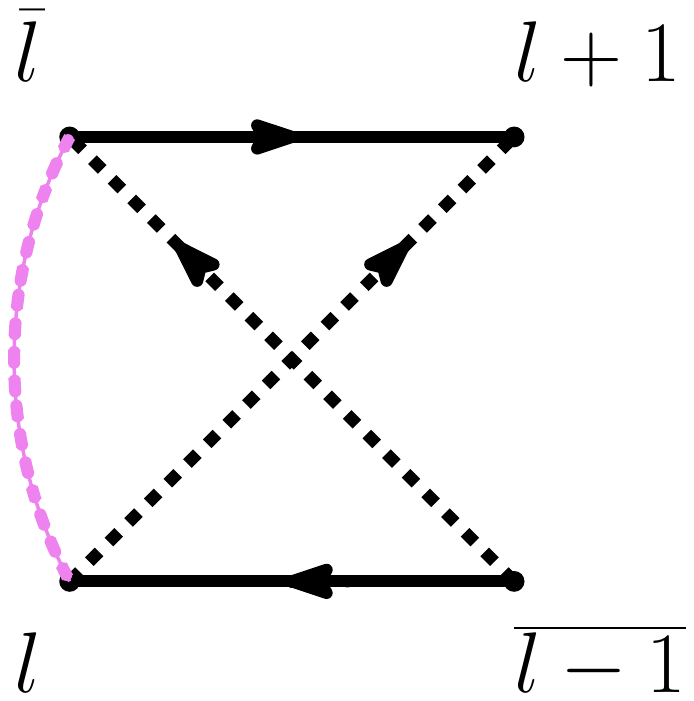}
      \caption{one-edge arc}
      \label{fig:tau identity and one equality}
    \end{subfigure}
    \caption{Completing the arc from $\cc{l}$ to $m$.}
    \label{fig:valid4cycle}
  \end{figure}

  Recall that $G_{\omega}$ must have a closed path traversing all
  dashed edges that alternates between directed and undirected edges
  where all directed edges point in the same direction along the
  path. Let us now consider the consequences of
  $[\omega,T]=\id$. Since $\omega(q)=\overline{l}$, there is an
  undirected dashed edge $(\overline{l} - \omega(\overline{q}))$. In
  order to produce no undirected chords,
  $(\overline{l},\omega(\overline{q}))$ must match with the existing
  undirected edge attached to $\overline{l}$, that is
  $(l,\overline{l})$. In other words, we must have
  $\omega(\overline{q})=l$. Since a directed dashed edge goes into
  $\overline{l}$, a directed dashed edge must also leave $l$. If this
  dashed directed edge follows the solid directed edge connected to
  $l$, the two edges will be directionally imbalanced, as in Figure
  \ref{fig:directed four-cycle with bad path}.  Thus there must be a
  directed chord leaving $l$, but since we already have the maximum
  allowed number of chords, we must have $m=l$, as shown in Figure
  \ref{fig:tau identity and one equality}.

  From this we see there can be only one isomorphism class of graphs
  in this set of permutations. Notice that the graph above can be
  produced by
  $\omega = (l \hspace{1mm} \overline{l})=s^l (1 \hspace{1mm} \cc
  1)s^{-l}$. Thus we will characterize our set of permutations as
  $\Irreg{\id}{2^1}=\{ s^n(1\,\cc{1})s^m:n,m\in [k] \}$. Since the
  permutations $s^n(1\,\cc{1})s^m$ are different for distinct values of
  $m$ and $n$, we may conclude that $|\Irreg{\id}{2^1}|=k^2$.
\end{proof}

\begin{lemma}
  \label{lem:irreg2_2}
  $\Irreg{2^1}{2^1}= \{ s^n(b \hspace{1mm} \cc 1)s^m
  \colon n,m\in [k], 3\leq b \leq k \}$,
  $\left|\Irreg{2^1}{2^1}\right|=k^2(k-2)$.
\end{lemma}

\begin{proof}
  We have already shown that having $\beta=2^1$ will give us a
  directed cycle of the type shown in Figure \ref{fig:directed
    four-cycle}. Since we now have $\alpha= 2^1$, we will also have an
  undirected 4-cycle.

 We will first consider the possibilities when we start with the specific directed 4-cycle type given in Figure \ref{fig:tau identity and one equality}, that is when $m=l$. Is it possible to have an undirected 4-cycle somewhere else in the diagram without creating more directed chords? Figures \ref{fig:undirected 4-cycle doesn't work 1} and \ref{fig:undirected 4-cycle doesn't work 2} show the only two possible types of undirected 4-cycles that could appear. We observe that in the first, a closed walk on dashed edges would not have all directed edges pointing the same way along the path and that in the second, we cannot have a closed walk on the dashed edges at all. Hence for a permutation from $\Irreg{2^1}{2^1}$, in Figure \ref{fig:directed four-cycle} there is at least one directed edge between $m$ and $\overline{l}$. In other words,  $m\neq l$.

\begin{figure}
\centering
\begin{subfigure}[b]{0.48\textwidth}
\centering
\includegraphics[scale=0.5]{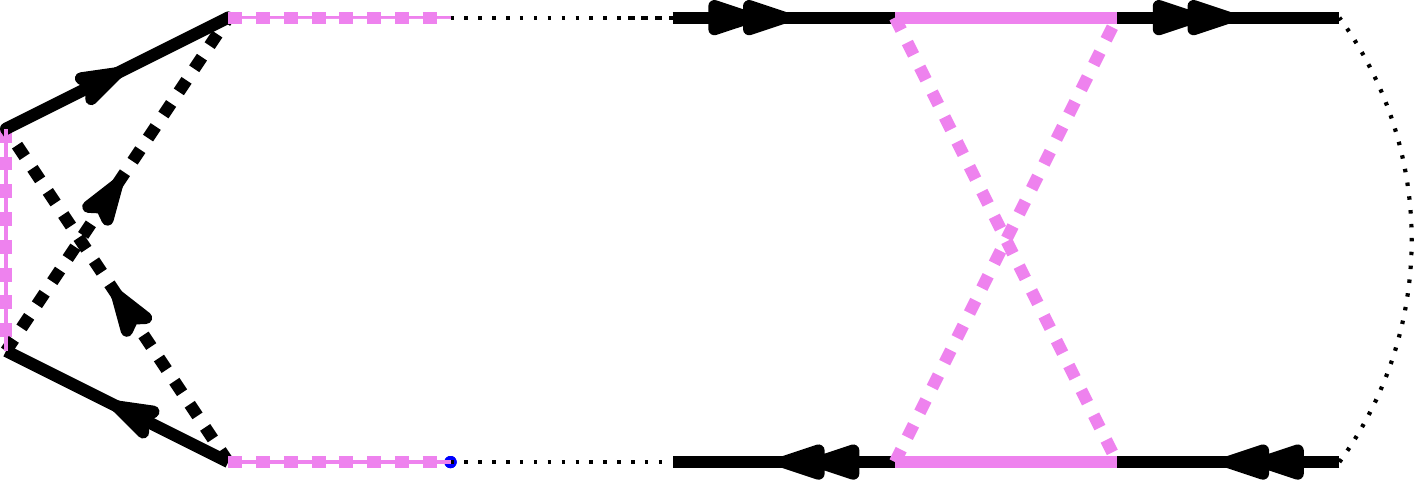}
\caption{wrong directions along dashed cycle}
\label{fig:undirected 4-cycle doesn't work 1}
\end{subfigure}
\hfill
\begin{subfigure}[b]{0.48\textwidth}
\centering
\includegraphics[scale=0.5]{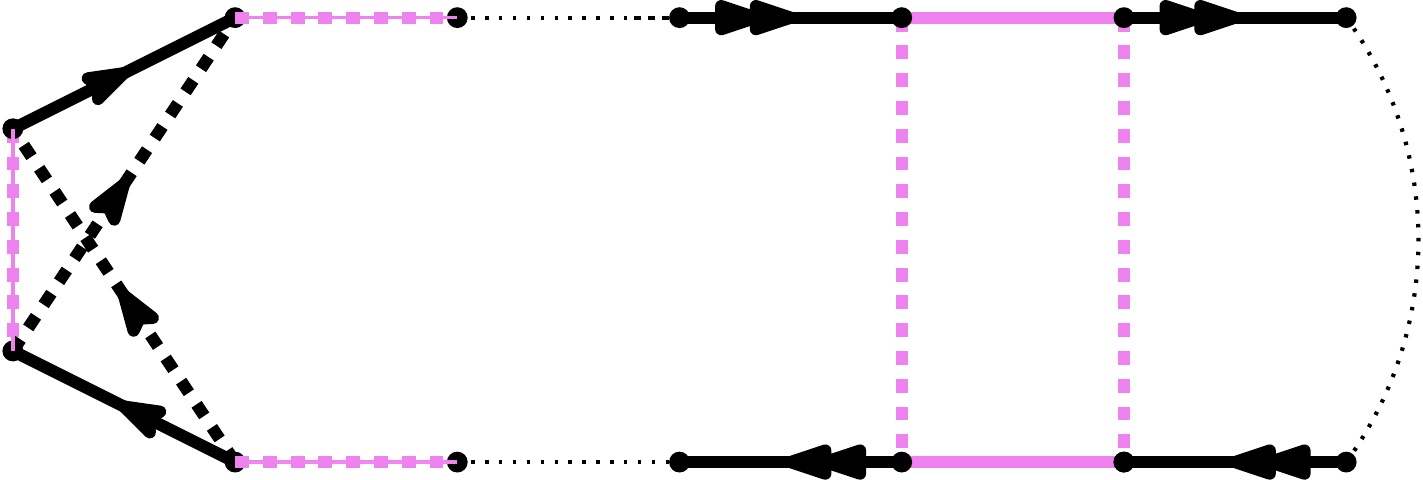}
\caption{disconnected dashed cycle}
\label{fig:undirected 4-cycle doesn't work 2}
\end{subfigure}
\caption{Invalid configurations from the proof of Lemma~\ref{lem:irreg2_2}.}
\end{figure}

Let us consider the edges to the left of the directed cycle, beginning with the dashed undirected edge incident to $\cc l$. If this edge is not a chord, we would have Figure \ref{fig:directed four-cycle with bad path} since no more directed edges can be chords either. We have already noted that the dashed path is impossible because of the unbalanced 2-cycle. Similarly, when assuming the undirected dashed edge incident to $m$ is not a chord, we obtain Figure \ref{fig:directed four-cycle with bad path 2} which also contains an illegal path. Thus, in order to have a permissible dashed path, we need to have undirected dashed chords attached to both $m$ and $\overline{l}$.

\begin{figure}
  \centering
  \includegraphics[scale=0.5]{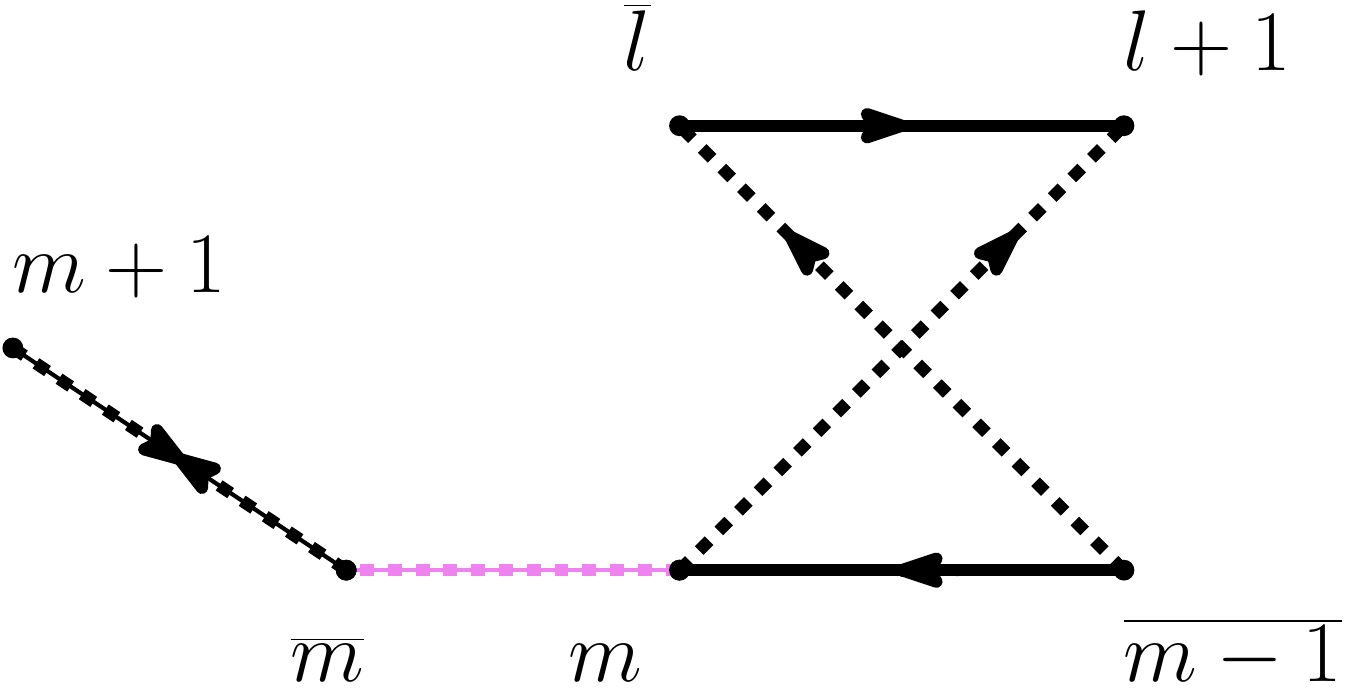}
  \caption{Impossible arc from $\cc{l}$ to $m$.}
  \label{fig:directed four-cycle with bad path 2}
\end{figure}

We now have two options: one is to have the dashed chord $(\overline{l},m)$, and the other is to have two different chords proceeding from $\overline{l}$ and $m$. The first option is shown in Figure \ref{fig:almost m l-bar structure}. Notice that we have only one way to finish the undirected 4-cycle, but this configuration admits no way of having only one closed dashed path without creating more chords, which we cannot have.  Thus the only possible completion of the undirected 4-cycle is the structure shown in Figure \ref{fig:m l-bar structure}, which corresponds to the permutation $\omega=(m \hspace{1mm} \overline{l})$ for $m \neq l$ and $m \neq l+1$. 

\begin{figure}
  \centering
  \begin{subfigure}[b]{0.45\textwidth}
    \centering
    \includegraphics[scale=0.5]{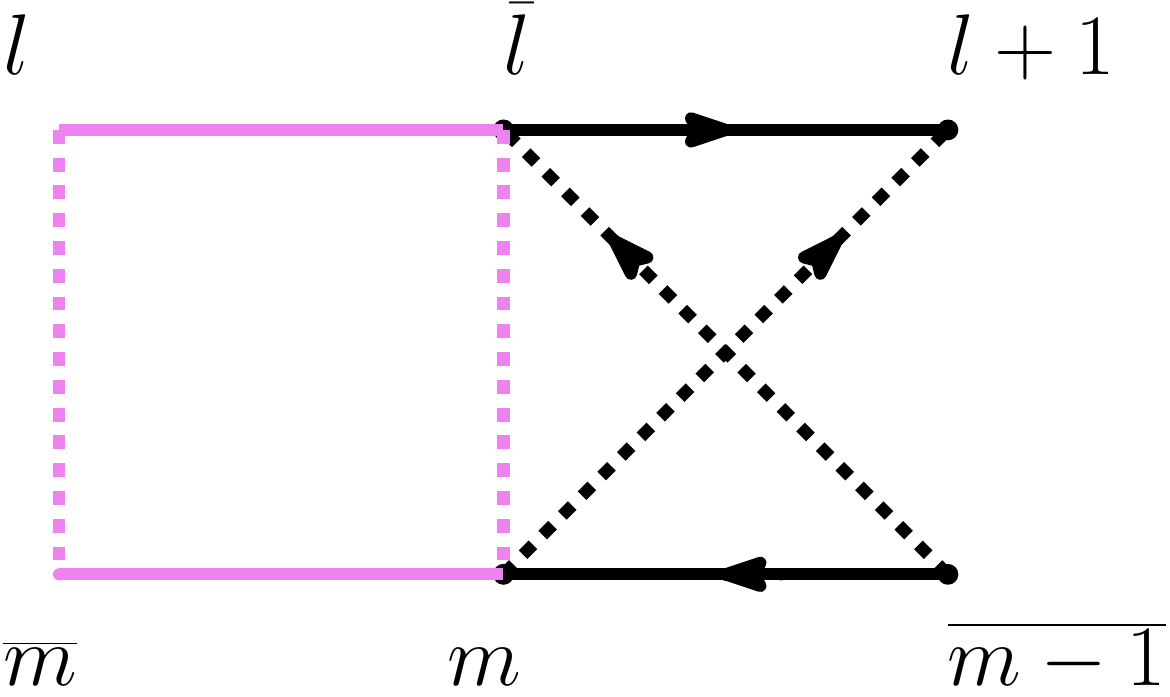}
    \caption{Dashed cycle is disconnected.}
    \label{fig:almost m l-bar structure}
  \end{subfigure}
  \begin{subfigure}[b]{0.45\textwidth}
    \centering
    \includegraphics[scale=0.5]{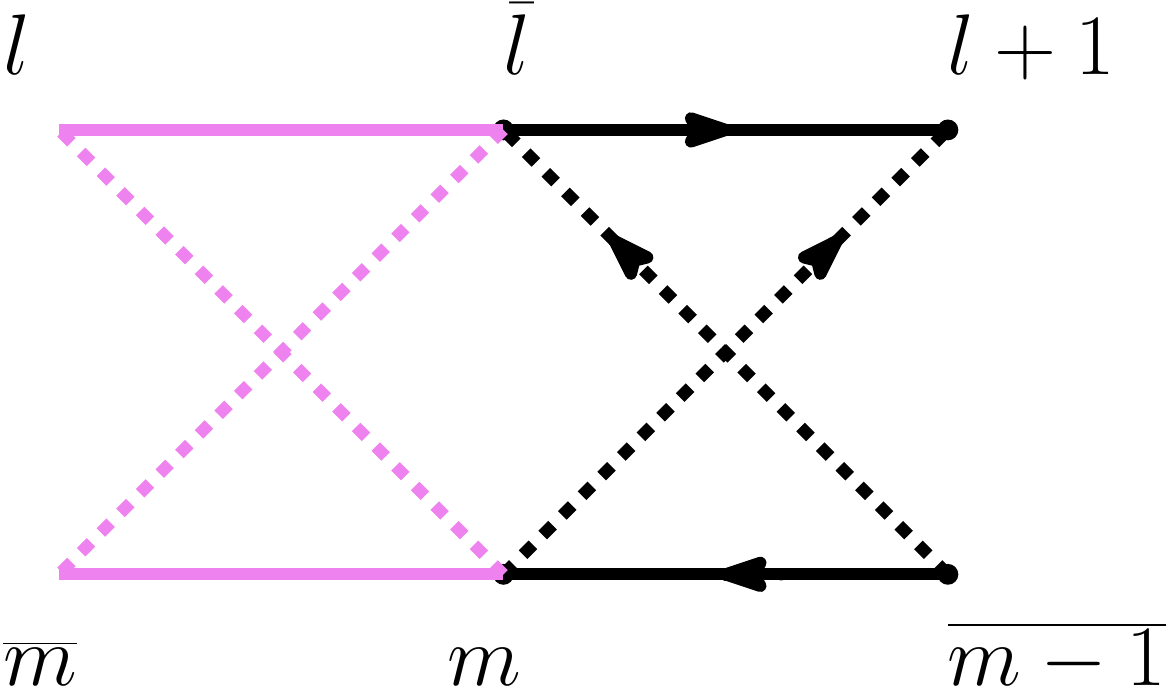}
    \caption{Valid configuration.}
    \label{fig:m l-bar structure}
  \end{subfigure}
  \caption{Adding the undirected cycle in the proof of
    Lemma~\ref{lem:irreg2_2}.}
\end{figure}

The different isomorphism classes can be represented by the graphs
produced by $(b\,\cc1)$ for $3 \leq b \leq k$. the permutations
$s^n(b\,\cc{1})s^m$ are different for distinct values of $b$, $m$ and
$n$, we may conclude that $|\Irreg{2^1}{2^1}|=k^2(k-2)$.
\end{proof}

\begin{lemma}
  \label{lem:irreg22}
  $\Irreg{\id}{2^2}= \{ s^n(1\, \cc b)(b \, \cc1)s^m,
  s^n(1 \, \cc 1)(b \, \cc b)s^m
  \colon n,m\in [k], 3\leq b \leq k-1  \}$,
  $|\Irreg{\id}{2^2}|=k^2(k-3)$.
\end{lemma}

\begin{proof}
Graphs for the permutations $\omega \in \Irreg{\id}{2^2}$ will have exactly two directed 4-cycles and no undirected chords. Since each $\omega$ is irregular, our previous work shows that at least one of the 4-cycles is of the form shown in Figure \ref{fig:directed four-cycle}. Let us assume for now that we do not have $m=l$, i.e. we do not have the case shown in Figure \ref{fig:tau identity and one equality}. Consider a dashed path going in the forward direction and starting with the directed edge $(\overline{m-1}\to \overline{l})$. Since we have no undirected chords, so the next edge in the path must be $(\overline{l},l)$. As shown in Figure \ref{fig:directed four-cycle with bad path}, we cannot continue on to the directed edge $(\overline{l-1}\to l)$ without creating an unbalanced cycle. Since $m\neq l$, we must have a third directed chord proceeding from $l$ that was not in our original 4-cycle. Similarly, following the dashed path backward from $(m \to l+1)$, we find there must be a fourth directed chord entering $\overline{m}$. This leaves only two possible configurations for what happens to the left of our original 4-cycle, as shown in Figure \ref{fig:almost two directed four-cycles} and Figure \ref{fig:two directed four-cycles}. However, upon attempting to complete the second directed 4-cycle we see the case in Figure \ref{fig:almost two directed four-cycles} would have an unbalanced directed cycle, making Figure  \ref{fig:two directed four-cycles} the only possible configuration, a graph corresponding to the permutation $(m \, \overline{l})(l \, \overline{m})$. Thus the permutations corresponding to graphs of this type are $\{ s^n(1 \, \cc b)(b \, \cc 1)s^m:n,m\in [k], 3\leq b \leq k-1  \}$.

\begin{figure}
\centering
\includegraphics[scale=0.5]{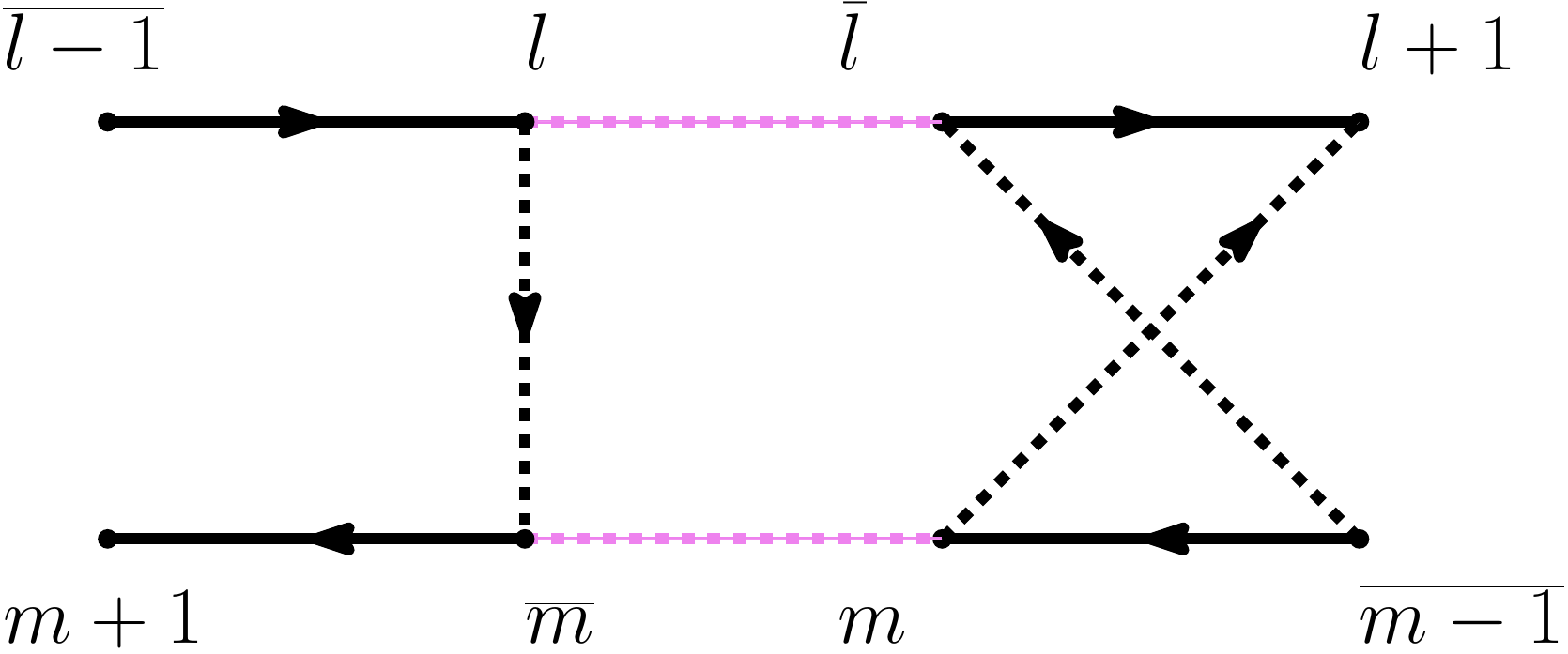}
\caption{Balanced 4-cycle cannot be formed.}
\label{fig:almost two directed four-cycles}
\end{figure}

\begin{figure}
\centering
\includegraphics[scale=0.5]{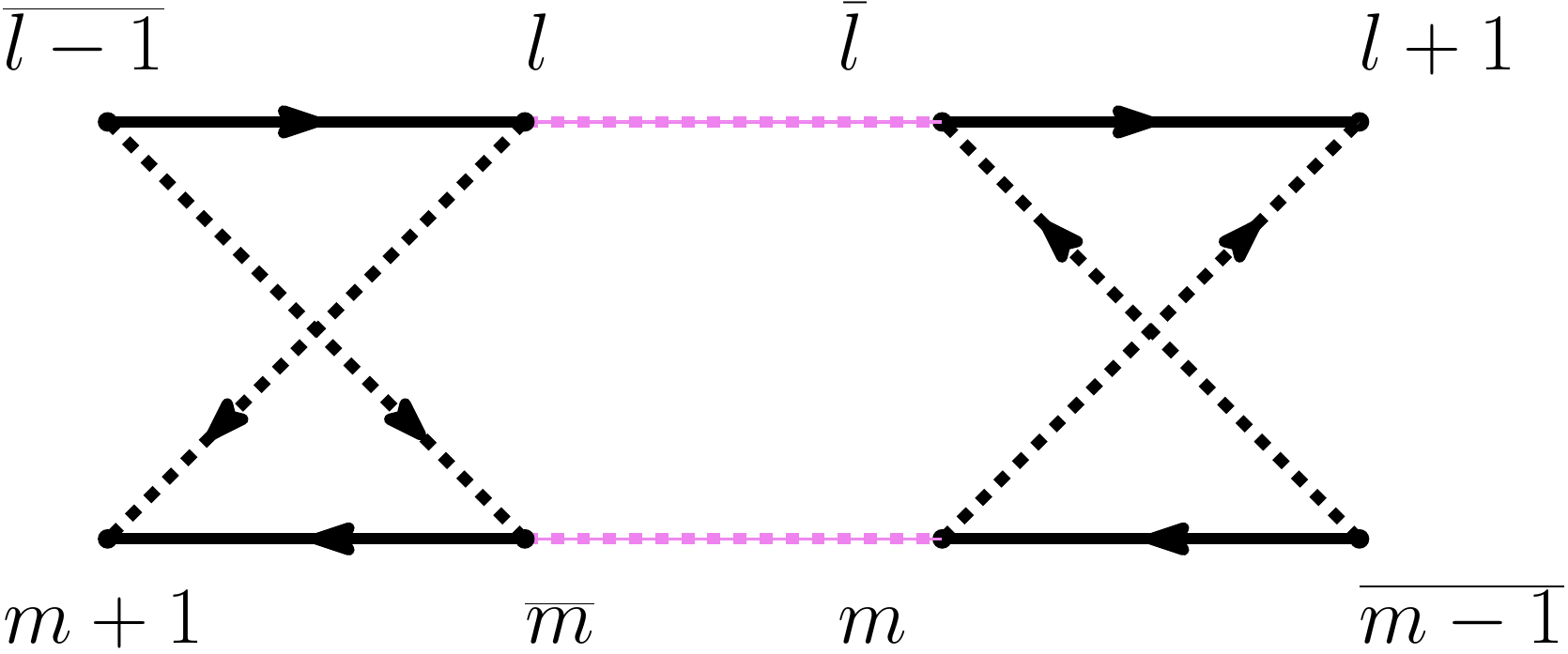}
\caption{A valid configuration for Lemma~\ref{lem:irreg22}}
\label{fig:two directed four-cycles}
\end{figure}

Now consider the case where $m=l$ in our first 4-cycle. Then it would have the form shown in Figure \ref{fig:tau identity and one equality}. The other 4-cycle can only take three forms, as shown in Figure \ref{fig:m=l 3 options}. Notice that in Figure \ref{fig:directed 4-cycle doesn't work 1}, we are unable to complete a closed walk along the dashed edges. In Figure \ref{fig:directed 4-cycle doesn't work 2} it is possible to complete a closed walk on the dashed edges, but not all directed edges will be pointing the same way. (Note that if the directions of the dashed edges in the second 4-cycle are inverted, the same problem occurs.) Thus the only remaining option is \ref{fig:two directed four-cycles 2}, where both 4-cycles have the same form.

\begin{figure}
\centering
\begin{subfigure}[b]{0.45\textwidth}
\centering
\includegraphics[width=2.75in]{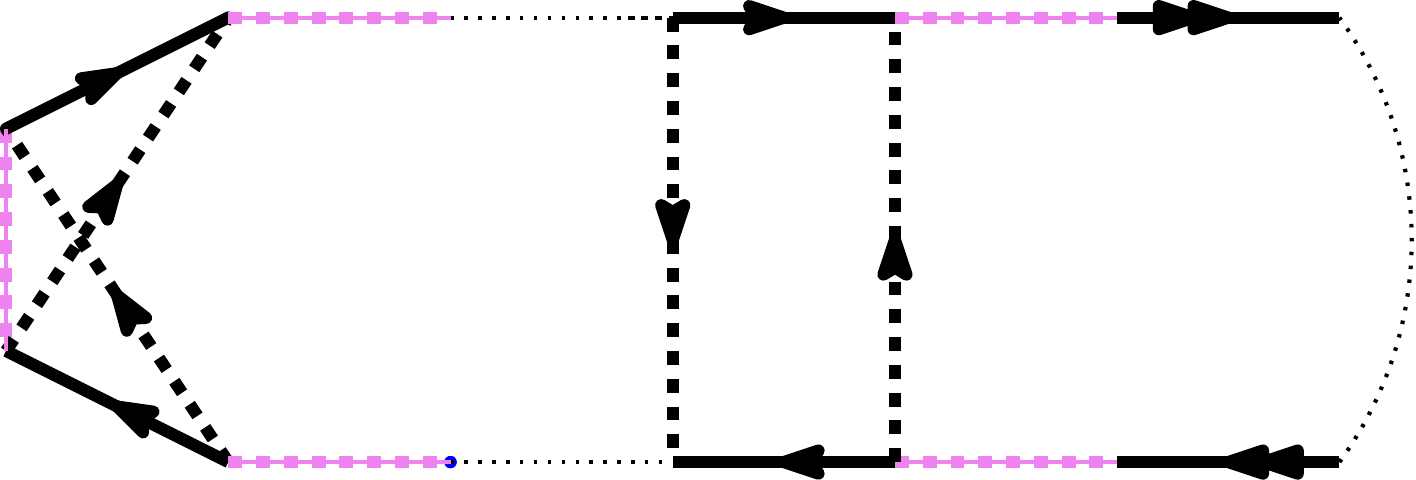}
\caption{disconnected dashed cycle}
\label{fig:directed 4-cycle doesn't work 1}
\end{subfigure}
\hfill
\begin{subfigure}[b]{0.45\textwidth}
\centering
\includegraphics[width=2.75in]{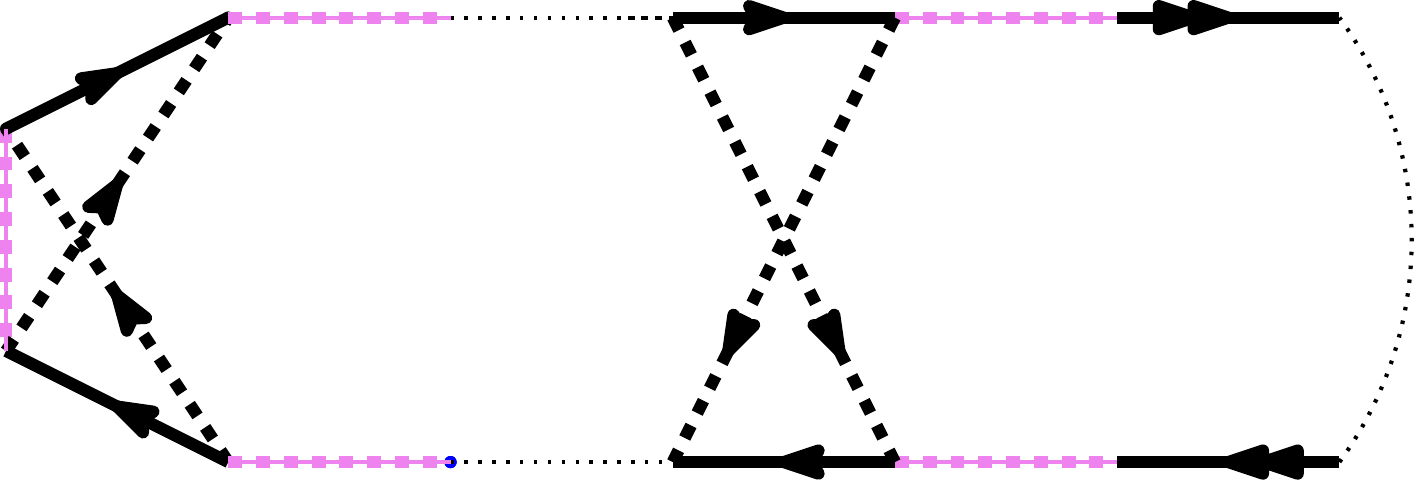}
\caption{wrong directions along dashed cycle}
\label{fig:directed 4-cycle doesn't work 2}
\end{subfigure}
\caption{Invalid configurations for Lemma~\ref{lem:irreg22}.}
\label{fig:m=l 3 options}
\end{figure}

\begin{figure}
\centering
\includegraphics[width=1.75in]{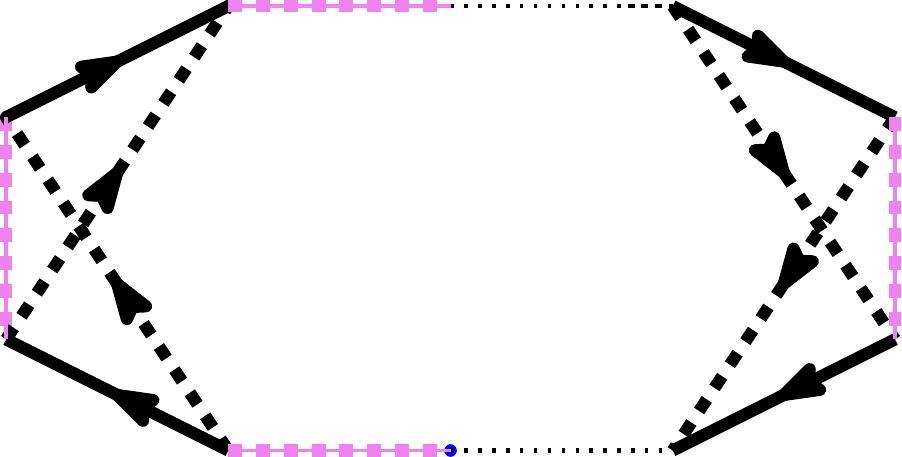}
\caption{Second valid configuration for Lemma~\ref{lem:irreg22}.}
\label{fig:two directed four-cycles 2}
\end{figure}

We already know that the individual cycles can be produced by $(l \, \cc l)$, so the complete set of permutations corresponding to graphs of this form will be $\{s^n(1 \, \cc 1)(b \, \cc b)s^m: n,m\in [k], 3\leq b\leq k -1 \}$.

Counting directly from our set notation, it would appear that
$|\Irreg{\id}{2^2}|=2k^2(k-3)$. However, we would be overcounting due
to symmetry. Notice that the transformation
\begin{align*}
  b &\mapsto k+2-b', \\
  n &\mapsto n'=n+b'-1, \\
  m &\mapsto m'+b'-1
\end{align*}
leaves the permutations $s^n (1 \, \cc b)(b \, \cc1)s^m$ and
$s^n(1 \, \cc 1)(b \, \cc b)s^m$ invariant.  Thus in the calculations
above we are counting each permutation exactly twice, and the actual
cardinality of the set is $k^2(k-3)$.
\end{proof}

\begin{lemma}
  \label{lem:irreg3}
  $\Irreg{\id}{3^1}=\{ s^n(1 \, \cc 1)(2 \, \cc 2)s^m
  \colon n,m \in [k]  \}$,
  $|\Irreg{\id}{3^1}|=k^2$.
\end{lemma}

\begin{proof}
  We are now considering the case of having exactly one directed
  6-cycle and no undirected cycles longer than 2. Since we are using
  irregular permutations $\omega$, we have
  $\omega(q)=\overline{l}$ for some $q,l\in [k]$. By Lemma B.2 we also
  have $\omega(\cc q)=l$. These imply we have a chord going into
  $\overline{l}$ and a chord coming out of $l$, as shown in Figure
  \ref{fig:6-cycle structure beginning}.

\begin{figure}
  \centering
  \begin{subfigure}[b]{0.45\textwidth}
    \centering
    \includegraphics[scale=0.6]{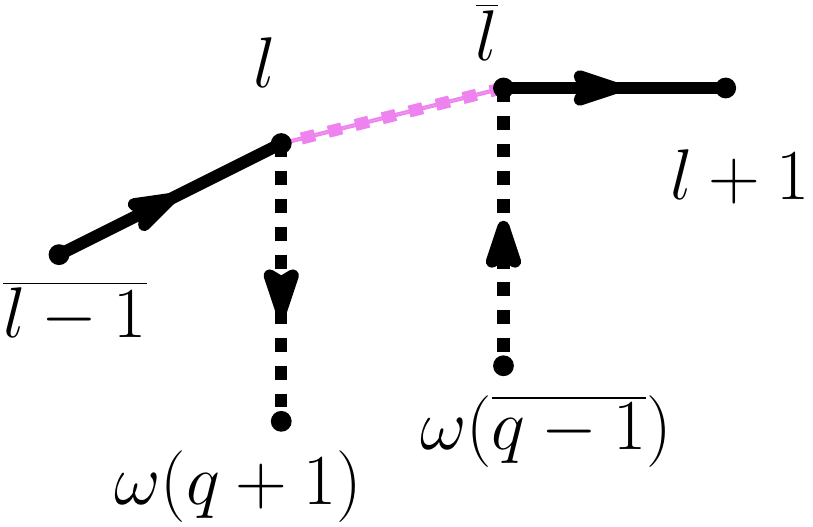}
    \caption{balanced 6-cycle cannot be formed}
    \label{fig:6-cycle structure beginning}
  \end{subfigure}
  \begin{subfigure}[b]{0.45\textwidth}
    \centering
    \includegraphics[scale=0.6]{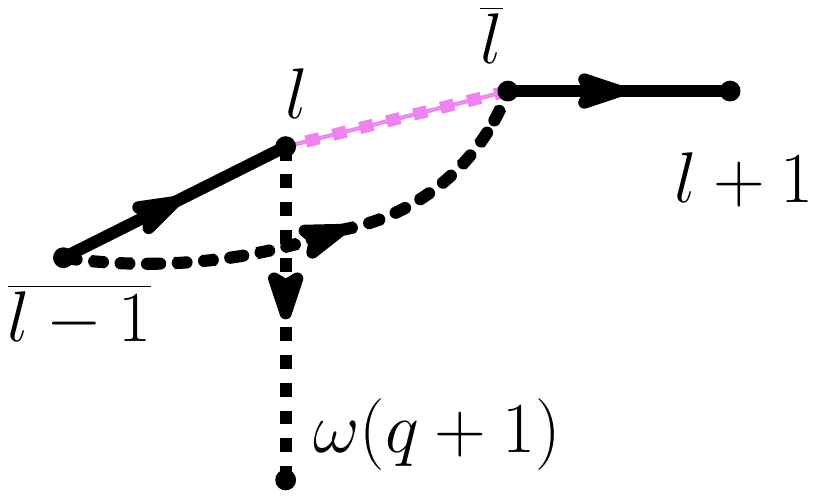}
    \caption{A valid configuration, case (1)}
    \label{fig:6-cycle structure case 1}
  \end{subfigure}
  \caption{A part of balanced 6-cycle.}
\end{figure}

These two chords must belong to our single directed 6-cycle, so we will examine the connection possibilities. Recall that our 6-cycle must alternate dashed and solid directed edges. One connection possibility is that $\omega(q+1)$ is connected to $\omega(\overline{q-1})$ by a single solid edge. However, no matter what direction we choose for this edge, we will have four edges pointing in the same direction with only two edges left with which to balance the directed cycle. We conclude this connection type is impossible. The same problem occurs if we attempt to connect $\overline{l-1}$ and $l+1$ with a dashed edge. Therefore, we are left with two cases: (1) $\overline{l-1}=\omega(\overline{q-1})$ or (2)  $\omega(q+1)=l+1$. The proofs for these cases are very similar, so we will only discuss the proof for (1) here. Observe our chosen connection for this case in Figure \ref{fig:6-cycle structure case 1}.

To finish the 6-cycle we need two more directed edges connecting $\omega(q+1)$ and $l+1$, one dashed and one solid. Now $\omega(q+1)$ must be either $r$ or $\overline{r}$ for some $r\in [k]$. 

The diagrams representing the two possibilities are as follows: Figure \ref{fig:w(q+1)=r} for $\omega(q+1)=r$ and Figure \ref{fig:w(q+1)=rbar} for  $\omega(q+1)=\overline{r}$. Notice that in the first case, the only way to finish the directed 6-cycle is with the dashed edge $(\cc {r-1} \to l+1)$. However, this creates two disjoint dashed paths, so this will not work.

 In the second case, the only way to complete the directed 6-cycle is with the dashed edge $(l+1\to r+1)$. Now consider the dashed path on this graph. We already have the maximum number of chords allowed. Since any dashed directed edges between $r$ and $l+1$ must coincide with solid directed edges, the only way to ensure all the directed edges on the dashed path point in the same direction is to eliminate the possibility of such edges and make $r=l+1$. This gives us the configuration shown in Figure \ref{fig:tau identity 6-cycle}. Note that this structure is produced by $\omega=(l \, \overline{l})$$(l+1 \hspace{2mm} \overline{l+1})$. Hence the set of all permutations that produce this structure will be $\{s^n(1 \, \cc 1)(2 \, \cc 2)s^m$ for $n,m\in [k]\}$. Since there is no rotational symmetry, we may conclude that the cardinality of the set is $k^2$.

\begin{figure}
\centering
\begin{subfigure}[b]{0.45\textwidth}
\centering
\includegraphics[width=2.5in]{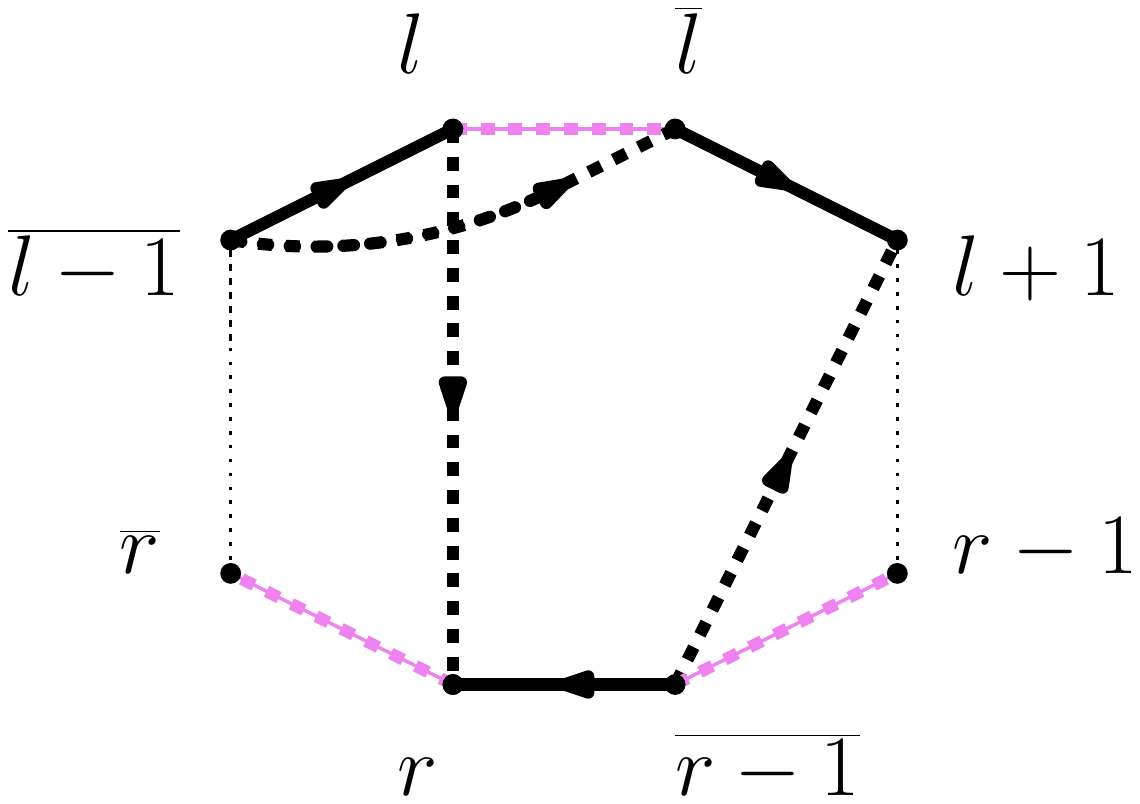}
\caption{disconnected dashed cycled}
\label{fig:w(q+1)=r}
\end{subfigure}
\hfill
\begin{subfigure}[b]{0.45\textwidth}
\centering
\includegraphics[width=2.5in]{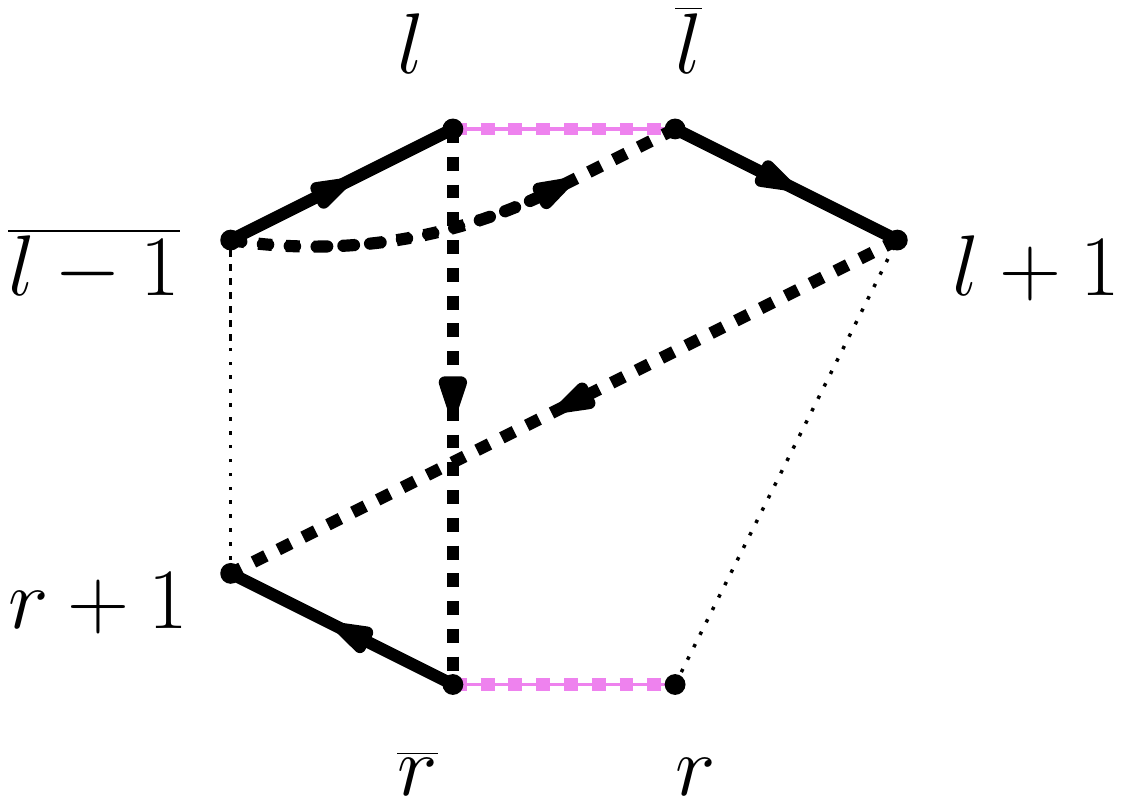}
\caption{arc $\cc{r}$ to $l+1$ must be short to prevent unbalanced cycles}
\label{fig:w(q+1)=rbar}
\end{subfigure}
\caption{Completing the 6-cycle.}
\end{figure}

\begin{figure}
\centering
\includegraphics[scale=0.6]{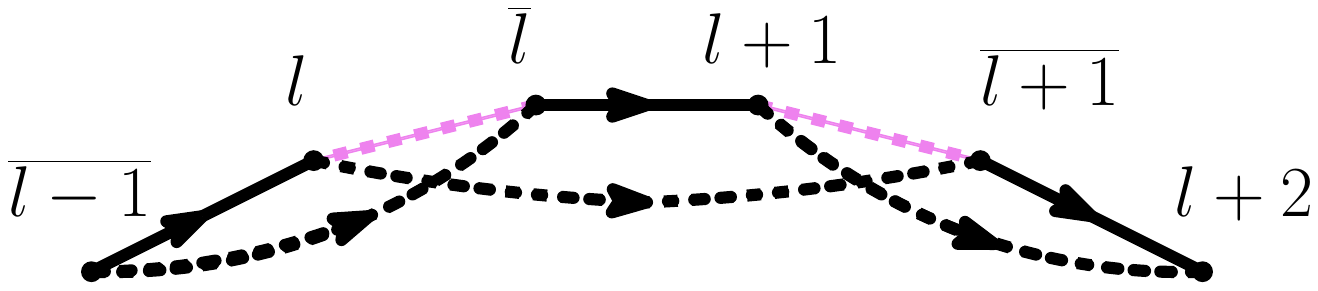}
\caption{The valid configuration for Lemma~\ref{lem:irreg3}.}
\label{fig:tau identity 6-cycle}
\end{figure}

\end{proof}

\bibliographystyle{myalpha}
\bibliography{twisted}

\end{document}